\documentclass[12pt]{iopart}
\expandafter\let\csname equation*\endcsname\relax
  \expandafter\let\csname endequation*\endcsname\relax
\usepackage{graphics} 
\usepackage{epsfig} 
\usepackage{mathptmx} 
\usepackage{times,yhmath} 
\usepackage{amssymb, amsmath, mathrsfs, multicol}

\usepackage{amsthm}
\usepackage{cite}
\usepackage[boxruled]{algorithm2e}
\SetAlFnt{\small}
\SetAlCapFnt{\small}
\SetAlCapNameFnt{\small}

\makeatletter
\renewcommand{\algocf@caption@boxruled}{%
  \hrule
  \hbox to \hsize{%
    \vrule\hskip-0.4pt
    \vbox{   
       \vskip\interspacetitleboxruled%
       \unhbox\algocf@capbox\hfill
       \vskip\interspacetitleboxruled
       }%
     \hskip-0.4pt\vrule%
   }\nointerlineskip%
}%
\makeatother
\usepackage{thmtools}
\usepackage{etoolbox}
\usepackage{subeqnarray}
\usepackage{epstopdf} 
\usepackage{cases}
\usepackage{tabularx}
\usepackage{phonetic}
\newcommand{\norm}[1]{\left\lVert#1\right\rVert}
\DeclareMathOperator{\rank}{rank} %
\DeclareMathOperator{\diag}{diag} %
\DeclareMathOperator{\im}{Im} %
\DeclareMathOperator{\re}{Re} %
\DeclareMathAlphabet{\mathpzc}{OT1}{pzc}{m}{it}
\allowdisplaybreaks
\usepackage{calc,pifont} 
\newcounter{local}
\renewcommand\theenumi{\protect\setcounter{local}%
  {171+\the\value{enumi}}\protect\ding{\value{local}}}

\theoremstyle{definition}
\newtheorem{thm}{Theorem}
\newtheorem{definition}{Definition}
\newtheorem{lem}{Lemma}
\newtheorem{rmk}{Remark}
\newtheorem{exam}{Example}

\newtheorem{prop}{Proposition}

\DeclareMathAlphabet{\mathpzc}{OT1}{pzc}{m}{it}

\begin{document}

\title[]{Pure Gaussian states from quantum harmonic oscillator chains with a single local dissipative process}

\author{Shan Ma$^1$,~Matthew J. Woolley$^1$,~Ian R. Petersen$^1$, and Naoki Yamamoto$^2$}

\address{$^1$School of Engineering and Information Technology, University of New South Wales 
at the Australian Defence Force Academy, Canberra, ACT 2600, Australia}
\address{$^2$Department of Applied Physics and Physico-Informatics, 
Keio University, Yokohama 223-8522, Japan}
\ead{shanma.adfa@gmail.com,  m.woolley@adfa.edu.au,  i.r.petersen@gmail.com, and yamamoto@appi.keio.ac.jp}
\vspace{10pt}

\begin{abstract}
We study the preparation of entangled pure Gaussian states via reservoir engineering. In particular, we consider a chain consisting of $(2\aleph+1)$ quantum harmonic oscillators where the central oscillator of the chain is coupled to a single reservoir. We then completely parametrize the class of  $(2\aleph+1)$-mode pure Gaussian states that can be prepared by this type of quantum harmonic  oscillator chain. This parametrization allows us to determine the steady-state entanglement properties of such quantum harmonic  oscillator chains. 
\end{abstract}
\begin{indented}
\item[]{\it Keywords}: Pure Gaussian states, Linear quantum systems, Reservoir engineering, Harmonic  oscillator chain, Nearest-neighbour Hamiltonian,  Local dissipation.
\end{indented}

\section{Introduction}
Gaussian states play an essential role in continuous-variable quantum information processing~\cite{BP03:book,AI07:jpa,weedbrook12:rmp}. Therefore, the preparation of pure Gaussian states is an important task~\cite{A06:prl}. Mathematically, any pure Gaussian state  can be prepared beginning with the vacuum state, and then applying a Gaussian unitary action whose Heisenberg action is a symplectic linear transformation on the vector of quadrature operators~\cite{SSM88:pra,SMD94:pra,A06:prl}. This method of pure Gaussian state preparation is a closed-system approach. Here we consider the preparation of pure Gaussian states via an open-system approach.   The main idea is that by engineering coherent and dissipative processes, a quantum system can be made strictly stable and will evolve into a given  pure Gaussian state.  
This approach is known as {\it reservoir engineering}~\cite{PCZ96:prl,WC14:pra,P04:jpa}. It is an efficient and robust approach to driving a quantum system into a desired target quantum state. In the finite-dimensional case,  the problem of pure quantum state stabilization by reservoir engineering has been studied theoretically  in~\cite{Y05:pra,KBDKMZ08:pra,VWC09:np}. 
In the infinite-dimensional case, the problem of preparing a pure Gaussian state  via reservoir engineering has recently been explored in~\cite{KY12:pra,Y12:ptrsa,IY13:pra,MWPY14:msc,MWPY16:arxiv,MPW15:arxiv}. In this paper, we focus on the preparation of pure Gaussian states via reservoir engineering.   
We consider an open quantum system, the time evolution of which is governed by   a Markovian Lindblad 
master equation~\cite{WM10:book}: 
\begin{align}
\label{MME}
        \frac{d}{d t}\hat{\rho} 
            &=-i[\hat{H},\; \hat{\rho}]
                +\sum\limits_{j=1}^{K}\left(
                       \hat{c}_{j}\hat{\rho} \hat{c}_{j}^*
                            -\frac{1}{2}\hat{c}_{j}^*\hat{c}_{j}\hat{\rho} 
                                -\frac{1}{2}\hat{\rho} \hat{c}_{j}^*\hat{c}_{j}\right),
\end{align}
where $\hat{\rho}$ is the density operator, $\hat{H}=\hat{H}^*$ 
is the  Hamiltonian operator, $\{\hat{c}_{j}\}$ is a set of Lindblad operators 
that represent the coupling of the system with its environment, and $K$ is the number of dissipative channels. For convenience, we collect all of the Lindblad operators into a vector  
$\hat{L}\triangleq\left[
\hat{c}_{1} \;\hat{c}_{2} \;\cdots \;\hat{c}_{K}
\right]^{\top}$, and we call $\hat{L}$ the \emph{coupling vector}. 
The Lindblad master equation~\eqref{MME} can typically be derived if the system is coupled weakly to a very large environment~\cite{BP02:book}. Under some circumstances,  the evolution described by the Lindblad 
master equation~\eqref{MME} will be strictly stable and will approach a time-independent (stationary) state, i.e.,  
$\lim_{t \rightarrow \infty}\hat{\rho}(t)=\hat{\rho}(\infty)$. 
Based on this fact, it has been shown in~\cite{KY12:pra,Y12:ptrsa} that 
any pure Gaussian state can be prepared in a dissipative quantum system by engineering a suitable pair of operators $\left(\hat{H},\;\hat{L}\right)$. Using the result developed in~\cite{KY12:pra,Y12:ptrsa}, it has been found that for many pure Gaussian states, the quantum systems generating them can be difficult to implement experimentally, mainly because either the Hamiltonian $\hat{H}$ or the coupling vector $\hat{L}$ has a nonlocal coupling structure. 

In this paper, we restrict our attention to a chain of $(2\aleph+1)$ quantum harmonic oscillators which are numbered from left to right as $1,\cdots,(2\aleph+1)$ with nearest-neighbour Hamiltonian interactions. The central oscillator of the chain is coupled to a single reservoir. More specifically, the quantum harmonic  oscillator chain we consider has two crucial features. (\textrm{i}) The Hamiltonian $\hat{H}$ is of the form 
$\hat{H}=\sum\limits_{j=1}^{2\aleph+1}\frac{\omega_{j}}{2}\left(\hat{q}_{j}^{2}+\hat{p}_{j}^{2} \right)+\sum\limits_{j=1}^{2\aleph}g_{j}\left(\hat{q}_{j}\hat{q}_{j+1}+\hat{p}_{j}\hat{p}_{j+1} \right)$, where $\omega_{j}\in \mathbb{R}$, $j=1,2,\cdots,2\aleph+1 $, and $g_{j}\in \mathbb{R}$,  $j=1,2,\cdots,2\aleph $.  This type of Hamiltonian describes a set of nearest-neighbour beam-splitter-like interactions. (\textrm{ii})
 Only the central (i.e., the $(\aleph+1)$th) oscillator of the chain is coupled to the reservoir.  That is, the coupling vector $\hat{L}$ reduces to a single Lindblad operator which is of the form $\hat{L}=c_{1}\hat{q}_{\aleph+1}+c_{2}\hat{p}_{\aleph+1}$, where $c_{1}\in \mathbb{C}$ and $c_{2} \in \mathbb{C}$. A quantum harmonic  oscillator chain subject to the above two constraints should be relatively easy to implement experimentally.  We then develop an exhaustive parametrization of all those pure Gaussian states that can be prepared by this type of quantum harmonic  oscillator chain. This  parametrization allows us to determine the entanglement properties  of the corresponding pure Gaussian states. For example, for the quantum chain considered in~\cite{ZLV15:pra}, oscillators located at equal distances, on the left and right, from the central one are entangled in pairs. However, using the parametrization developed in this paper, we can find pure Gaussian states for which any two oscillators (except the central oscillator) in the chain are entangled.  Note that a chain structure of quantum harmonic oscillators has also been studied in, e.g.,~ \cite{AEP02:pra,BR04:pra}. 
 
It is worth remarking that although in this paper, we only consider the case where the chain has an odd number of quantum harmonic  oscillators, the  method developed here can be easily extended to handle the case where the number of the oscillators is even, as 
we have previously done in~\cite{MPW15:arxiv,MWPY16:cdc}. The method developed in this paper can also be easily extended to handle the case where the reservoir acts locally on an arbitrary oscillator  in the chain (not necessarily the central one). The parametrizations of pure Gaussian steady states in these cases involve a similar method of analysis to the method here, and hence are omitted.    
 
 The paper is organized as follows. In Section~\ref{Preliminaries}, we summarize some basic concepts and results on pure Gaussian states. In Section~\ref{constraints}, we define the type of open quantum harmonic  oscillator chain under consideration. Section~\ref{parametrization} and Section~\ref{Algorithm} contain the main result of this paper.  In Section~\ref{parametrization}, we characterize all stationary pure Gaussian states that can be prepared by a quantum harmonic  oscillator chain with a single reservoir acting locally on the central oscillator of the chain. This characterization is formulated as a theorem. To make the result more accessible to the reader, we also provide an equivalent algorithm in Section~\ref{Algorithm} for finding such pure Gaussian states. Applying the algorithm generates pure Gaussian state covariance matrices. The algorithm also enables us to determine the steady-state entanglement properties  of the quantum harmonic  oscillator chains.  The proof of the main theorem is left to the Appendix.   

\textit{Notation.} 
We use $\mathbb{R}$ to denote  the set of real numbers and $\mathbb{C}$ to denote the set of complex numbers. The set of real $m\times n$ matrices is denoted $\mathbb{R}^{m \times n}$, and  the set of complex-entried 
$m\times n$ matrices is denoted $\mathbb{C}^{m \times n}$. $I_{n}$ is the $n\times n$ identity matrix. $0_{m\times n}$ is the $m\times n$ zero matrix. $\norm{\cdot}$ denotes
  the Euclidean norm ($l_{2}$-norm) of a vector.  The superscript ${}^{\ast}$ denotes either 
the complex conjugate of a complex number or the adjoint of an operator.
For a matrix $A=[A_{jk}] \in \mathbb{C}^{m \times n}$, $A^{\top}=[A_{kj}]$ denotes the transpose of $A$, and $A^{\dagger}=[A_{kj}^{\ast}]$ denotes the complex conjugate transpose of $A$. For a matrix $A=[A_{jk}]$ with operator-valued entries, $A^{\top}=[A_{kj}]$ denotes the transpose of $A$, and $A^{\dagger}=[A_{kj}^{\ast}]$ denotes the transpose of $A$ with its elements replaced by the corresponding adjoint operators. For a real symmetric matrix $A=A^{\top}\in \mathbb{R}^{n \times n}$, $A>0$ means that $A$ is positive definite. We denote by $\diag[A_{1},\cdots,A_{n}]$ the block diagonal matrix whose diagonal blocks are $A_{j}$, $j=1,2,\cdots,n$. $\det(A)$ denotes the determinant of the matrix $A$.

\section{Preliminaries} \label{Preliminaries}
We consider a continuous-variable quantum system consisting of $N$ canonical bosonic modes. Suppose $\hat{q}_{j}$ and $ \hat{p}_{j}$ are the position and momentum operators for the $j$th mode, respectively. In particular, these operators satisfy the following commutation relations (we use $\hbar=1$ throughout the paper)
\begin{align*}
\left[\hat{q}_{j}, \hat{p}_{k}\right]=i\delta_{jk}, \quad \left[\hat{q}_{j}, \hat{q}_{k}\right]=0,\quad \text{and}\;\; \left[\hat{p}_{j}, \hat{p}_{k}\right]=0. 
\end{align*}
It is convenient to arrange the self-adjoint operators $\hat{q}_{j}$, $ \hat{p}_{j}$ into a column vector 
$\hat{x}=\left[\hat{q}_{1}\;\cdots\;\hat{q}_{N}\;\; \hat{p}_{1}\;\cdots\;\hat{p}_{N}\right]^{\top}$.  Then the commutation relations can be written as 
$\left[\hat{x}_{j}, \hat{x}_{k}\right]=i\Sigma_{jk}$, 
where $\Sigma_{jk}$ is the $(j,k)$ element of  the  matrix 
$\Sigma = \begin{bmatrix}
         0 & I_{N}\\
-I_{N} &0
\end{bmatrix}$.

Let $\hat{\rho}$ be the density operator of the system. 
Then the mean value of the  
vector $\hat{x} $ is given by 
$\langle \hat{x} \rangle 
=\left[\tr(\hat{q}_{1}\hat{\rho})\;\cdots\;\tr(\hat{q}_{N}\hat{\rho})
\; \tr(\hat{p}_{1}\hat{\rho})\;\cdots\;\tr(\hat{p}_{N}\hat{\rho})\right]^{\top}$ 
and the covariance matrix of the  
vector $\hat{x} $  is given by   $V=\frac{1}{2}\langle \triangle\hat{x}{\triangle\hat{x}}^{\top}+(\triangle\hat{x}{\triangle\hat{x}}^{\top})^{\top} \rangle$, where $\triangle\hat{x}=\hat{x}-\langle \hat{x}\rangle$. 
A Gaussian state is entirely characterized by its  mean vector $\langle \hat{x}\rangle$ 
and its covariance matrix $V$. Because the mean vector $\langle \hat{x}\rangle$ 
contains no information about noise and entanglement, it is irrelevant for our purpose and will be set to zero without loss of generality.  The purity of a Gaussian state is given  by $p=\tr(\hat{\rho}^{2})=2^{-N}/\sqrt{\det(V)}$. A Gaussian state with  covariance matrix $V$ is pure  if and only if $\det(V)=2^{-2N}$. 

The covariance matrix of a pure Gaussian state is a real and symmetric matrix which must satisfy $V+\frac{i}{2}\Sigma\ge 0$. It then follows that $V>0$~\cite{BF06:JPA,SMD94:pra,PSL09:pra,LS15:jpa}.  However, not all real, positive definite matrices correspond to the covariance matrix of a pure Gaussian state. 
If a matrix $V$ corresponds to the covariance matrix of an $N$-mode pure Gaussian state, it can always be decomposed as 
\begin{align}\label{covariance}
V=\frac{1}{2}\begin{bmatrix}
Y^{-1} &Y^{-1}X\\
XY^{-1} &XY^{-1}X+Y 
\end{bmatrix},
\end{align}
where $X=X^{\top}\in \mathbb{R}^{N \times N}$, $Y=Y^{\top}\in \mathbb{R}^{N \times N}$ and $Y>0$~\cite{MFL11:pra}. For example, the covariance matrix $V$ of  the $N$-mode vacuum state is given by $V=\frac{1}{2}I_{2N}$. In this case, using~\eqref{covariance}, we obtain $X=0_{N\times N}$ and $Y=I_{N}$. Let us define $Z\triangleq X+iY$. Given the matrix $Z$, a covariance matrix can be constructed from the real part $X$ and the imaginary part $Y$ of $Z$ using~\eqref{covariance}. Thus, the matrix $Z$ uniquely characterizes a pure Gaussian state. We refer to $Z=X+iY$ as the  \emph{Gaussian graph matrix}~\cite{MFL11:pra}. Note that, to ensure that the corresponding state is physical, the Gaussian graph matrix $Z$ must satisfy $\re(Z)=\re(Z)^{\top}$ and $\im(Z)=\im(Z)^{\top}>0$.

Suppose that the system Hamiltonian in~\eqref{MME} is quadratic in the quadrature operators, i.e., 
$\hat{H}=\frac{1}{2}\hat{x}^{\top}G\hat{x}$,
with $G=G^{\top}\in \mathbb{R}^{2N \times 2N}$, the coupling vector is  linear in the quadrature operators, i.e., 
$\hat{L} = C \hat{x}$,
with $C\in \mathbb{C}^{K \times 2N}$, and the dynamics of the density operator $\hat{\rho}$ 
obey the Markovian Lindblad master equation~\eqref{MME}.   Then from~\eqref{MME}, we can obtain the following dynamical equations for  
the mean vector $\langle \hat{x}(t) \rangle$ and the covariance matrix $V(t)$ of the canonical operators: 
\begin{numcases}{}
\frac{d\langle\hat{x}(t)\rangle}{dt}=\mathcal{A}\langle\hat{x}(t)\rangle, \label{meanfunction} \\
\frac{dV(t)}{dt}=\mathcal{A}V(t)+V(t)\mathcal{A}^{\top}+\mathcal{D},  \label{covfunction}
\end{numcases}
where $\mathcal{A}=\Sigma\left(G+\im(C^{\dagger}C)\right)$ and  
$\mathcal{D}=\Sigma\re(C^{\dagger}C)\Sigma^{\top}$ are referred to as drift and diffusion matrices, respectively~\cite{WD05:prl}, \cite[Chapter 6]{WM10:book}. 
The linearity of the dynamics guarantees that if the  system is initially prepared in a Gaussian state, then the system will maintain this  Gaussian character, with the mean vector $\langle \hat{x}(t) \rangle$ 
and the covariance matrix $V(t)$ evolving according to~\eqref{meanfunction} and~\eqref{covfunction}, respectively. 
We shall be particularly interested in the unique steady state of the master
equation~\eqref{MME} with the covariance matrix $V(\infty)$.
 Recently, a necessary and sufficient condition has been obtained in~\cite{KY12:pra,Y12:ptrsa} for preparing an arbitrary pure Gaussian steady state via reservoir engineering.  The result is summarized in the following Lemma. 

\begin{lem}[\cite{KY12:pra,Y12:ptrsa}]\label{lemnaoki}
Let $Z=X+iY$ be the Gaussian graph matrix of an $N$-mode pure Gaussian state. Then this pure Gaussian state is  the steady state of the master
equation~\eqref{MME} if and only if
\begin{align} \label{G}
G=\begin{bmatrix}
XRX+YRY-\Gamma Y^{-1}X-XY^{-1}\Gamma^{\top} &-XR+\Gamma Y^{-1}\\
-RX+Y^{-1}\Gamma^{\top} &R
\end{bmatrix},
\end{align}
 and 
\begin{align} \label{C}
C=P^{\top}\left[-Z\;\; I_{N}\right], 
\end{align} 
where  $R=R^{\top}\in\mathbb{R}^{ N\times N}$,  $\Gamma=-\Gamma^{\top}\in\mathbb{R}^{ N\times N}$, and  $P\in \mathbb{C}^{N\times K}$ are free matrices satisfying the following rank condition
\begin{align}
\rank\left([P\;\;\;QP\;\;\;\cdots\;\;\;Q^{N-1}P]\right)=N,\;\;  Q\triangleq -iRY+Y^{-1}\Gamma. \label{rankconstraint}
\end{align}
\end{lem}

\begin{rmk}
A pair $\left(A_{1},\;A_{2}\right)$ where $A_{1}\in \mathbb{C}^{n\times n}$ and $A_{2}\in \mathbb{C}^{n\times m}$ is said to be \emph{controllable} if the matrix $[A_{2}\;\;\;A_{1}A_{2}\;\;\;\cdots\;\;\;A_{1}^{n-1}A_{2}]$ has full row rank~\cite{ZJK96:book}. It then follows
that the rank condition~\eqref{rankconstraint} is equivalent to $\left(Q,\;P\right)$ being controllable. 
\end{rmk}

\begin{rmk}
The rank condition~\eqref{rankconstraint} guarantees the strict stability of the resulting linear quantum system with $\hat{H}=\frac{1}{2}\hat{x}^{\top}G\hat{x}$ and $\hat{L}=C\hat{x}$; see~\cite{Y12:ptrsa,MWPY16:arxiv} for details. 
\end{rmk}

\section{Constraints} \label{constraints}
In the sequel, we restrict our consideration to a special class of linear open quantum systems, i.e., quantum harmonic  oscillator chains subject to constraints.   Then in Section~\ref{parametrization}, we will investigate which pure Gaussian states can be prepared by this type of quantum harmonic  oscillator chain. The system we consider is a  chain consisting of $(2\aleph+1)$ harmonic oscillators, labelled $1$ to $(2\aleph+1)$ from left to right, subject to the following two constraints.    
\begin{enumerate}
\item The Hamiltonian $\hat{H}$ is of the form 
$\hat{H}=\sum\limits_{j=1}^{2\aleph+1}\frac{\omega_{j}}{2}\left(\hat{q}_{j}^{2}+\hat{p}_{j}^{2} \right)+\sum\limits_{j=1}^{2\aleph}g_{j}\left(\hat{q}_{j}\hat{q}_{j+1}+\hat{p}_{j}\hat{p}_{j+1} \right)$, where $\omega_{j}\in \mathbb{R}$, $j=1,2,\cdots,2\aleph+1 $, and $g_{j}\in \mathbb{R}$,  $j=1,2,\cdots,2\aleph $. \label{constraint1}
\item Only the central  oscillator of the chain is coupled to the reservoir. That is, the coupling vector $\hat{L}$ is of the form $\hat{L}=c_{1}\hat{q}_{\aleph+1}+c_{2}\hat{p}_{\aleph+1}$, where $c_{1} \in \mathbb{C}$ and $ c_{2} \in \mathbb{C}$. \label{constraint2}
\end{enumerate}

\begin{rmk}
The structure of the linear quantum system subject to the constraints~\ref{constraint1} and~\ref{constraint2} is shown in Fig.~\ref{fig1}. The system is a chain composed of $(2\aleph+1)$ quantum harmonic oscillators with nearest–-neighbour Hamiltonian interactions. Only the central (i.e., $(\aleph+1)$th) oscillator of the chain is coupled to the reservoir. The Hamiltonian described in the constraint~\ref{constraint1} can be rewritten in terms of annihilation and creation operators as
\begin{align}
\hat{H}&=\sum\limits_{j=1}^{2\aleph+1}\frac{\omega_{j}}{2}\left(\hat{q}_{j}^{2}+\hat{p}_{j}^{2} \right)+\sum\limits_{j=1}^{2\aleph}g_{j}\left(\hat{q}_{j}\hat{q}_{j+1}+\hat{p}_{j}\hat{p}_{j+1} \right) \notag\\
&=\sum\limits_{j=1}^{2\aleph+1}\frac{\omega_{j}}{2}\left(\hat{a}_{j}^{\ast}\hat{a}_{j}+\hat{a}_{j}\hat{a}_{j}^{\ast} \right)+\sum\limits_{j=1}^{2\aleph}g_{j}\left(\hat{a}_{j}^{\ast}\hat{a}_{j+1}+\hat{a}_{j}\hat{a}_{j+1}^{\ast} \right), \notag\\
&=\sum\limits_{j=1}^{2\aleph+1}\frac{\omega_{j}}{2}\left(2\hat{a}_{j}^{\ast}\hat{a}_{j}+1\right)+\sum\limits_{j=1}^{2\aleph}g_{j}\left(\hat{a}_{j}^{\ast}\hat{a}_{j+1}+\hat{a}_{j}\hat{a}_{j+1}^{\ast} \right),\notag\\
&\cong\sum\limits_{j=1}^{2\aleph+1}\omega_{j} \hat{a}_{j}^{\ast}\hat{a}_{j}+\sum\limits_{j=1}^{2\aleph}g_{j}\left(\hat{a}_{j}^{\ast}\hat{a}_{j+1}+\hat{a}_{j}\hat{a}_{j+1}^{\ast} \right), \label{beam_interaction}
\end{align}
where $\hat{a}_{j}=\frac{\hat{q}_{j}+i\hat{p}_{j}}{\sqrt{2}}$ and $\hat{a}_{j}^{\ast}=\frac{\hat{q}_{j}-i\hat{p}_{j}}{\sqrt{2}}$ are the annihilation and creation operators for the $j$th oscillator, respectively. The last relation follows from the fact that a constant term in the Hamiltonian does not produce any dynamics, and hence can be dropped. 
It can be seen immediately from~\eqref{beam_interaction} that the nearest–-neighbour Hamiltonian coupling is a beam-splitter-like interaction. Note that in the constraint~\ref{constraint1}, we require only that the parameters $\omega_{j}$, $j=1,2,\cdots,2\aleph+1 $, and $g_{j}$,  $j=1,2,\cdots,2\aleph $, are real. These parameters do not necessarily have the same or opposite values. Thus, the linear quantum harmonic  oscillator chain subject to the constraints~\ref{constraint1} and~\ref{constraint2}  is more general than the system studied in~\cite{ZLV15:pra}, where some symmetries and antisymmetries are assumed within the parameters $\omega_{j}$, $j=1,2,\cdots,2\aleph+1 $, and $g_{j}$,  $j=1,2,\cdots,2\aleph $. 
\begin{figure}[htbp]
\begin{center}
\includegraphics[height=2.5cm]{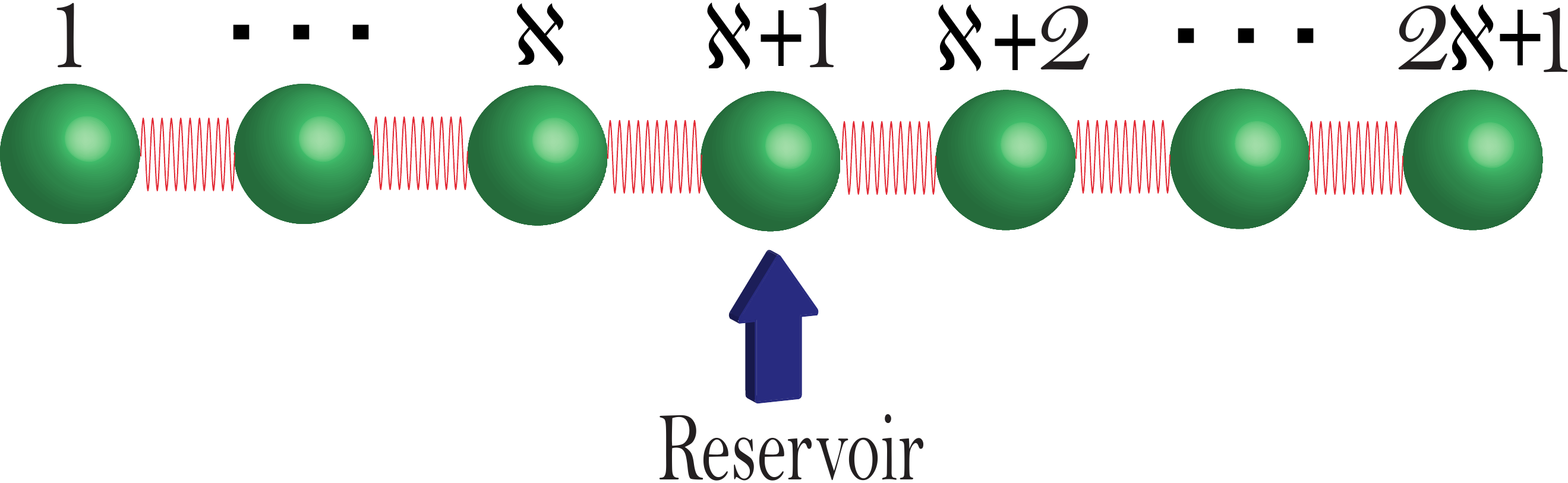}
\caption{
A chain consisting of $(2\aleph+1)$ quantum harmonic oscillators with nearest–-neighbour Hamiltonian interactions.  Only the central oscillator of the chain is coupled to the reservoir.}
\label{fig1}
\end{center}
\end{figure}
\end{rmk}

\begin{prop} \label{rmkvacuum}
The $(2\aleph+1)$-mode vacuum state, as a trivial pure Gaussian state, can be prepared by a quantum harmonic  oscillator chain subject to the two constraints~\ref{constraint1} and \ref{constraint2}.
\end{prop}
\begin{proof}
We prove this result by construction. We choose  $P=\begin{bmatrix} 0_{1 \times \aleph } &i &0_{ 1 \times \aleph} \end{bmatrix}^{\top}$,  $\Gamma=0$, and 
$
R=\begin{bmatrix}
-2 &1  &       & &&&&           &0\\
1  &-2 &\ddots \\
           &\ddots    &\ddots  &1\\
          &          &1 &-2 &1 \\
&&&1 &0 &1  &       &\\ 
&&&&1& 2 &1  &       &\\
&&&&&1  &2 &\ddots \\
& &&&&          &\ddots    &\ddots  &1\\
0&&&&&          &          &1 &2 
\end{bmatrix}\in\mathbb{R}^{ (2\aleph+1)\times (2\aleph+1)}$ in~\eqref{G} and~\eqref{C}. We next show  that the resulting quantum system with Hamiltonian $\hat{H}=\frac{1}{2}\hat{x}^{\top}G\hat{x}$ and  coupling vector $\hat{L}=C\hat{x}$ satisfies the constraints~\ref{constraint1} and \ref{constraint2} and generates the $(2\aleph+1)$-mode vacuum state.   Recall that the Gaussian graph matrix $Z$ corresponding to the $(2\aleph+1)$-mode vacuum state is given by $Z=iI_{2\aleph+1}$. Therefore, we have $X=0$ and $Y=I_{2\aleph+1}$. We need to show that the rank constraint~\eqref{rankconstraint} holds with the chosen matrices $P$, $\Gamma$ and $R$. That is, we need to show $(Q,\;P)$ is controllable. Since $Q=-iR$, it suffices to  show $(R,\;P)$ is controllable.  Let $\mathcal{P}_{1}\triangleq\begin{bmatrix}
0_{1\times \aleph} &1 &0_{1\times \aleph}\\
I_{\aleph} &0_{\aleph\times 1} &0_{\aleph\times \aleph}\\
0_{\aleph\times \aleph} &0_{\aleph\times 1} &I_{\aleph}
\end{bmatrix}$. Then we have $\mathcal{P}_{1} P= \begin{bmatrix}i \\ 0_{2\aleph \times 1} \end{bmatrix}$ and $\mathcal{P}_{1} R \mathcal{P}_{1}^{\top}=
\begin{bmatrix}
0 &\breve{R}_{21}^{\top}  \\
\breve{R}_{21}     &\breve{R}_{22}   
\end{bmatrix}$,
  where  $\breve{R}_{21} =\begin{bmatrix}
0_{(\aleph-1) \times 1}\\1 \\   1\\ 0_{(\aleph-1)\times 1}
\end{bmatrix}$, $
\breve{R}_{22}=\left[ {\begin{array}{*{20}c}
\breve{R}_{22,1} &\vline &0_{\aleph\times \aleph}   \\
\hline
0_{\aleph\times \aleph} &\vline & \breve{R}_{22,2}
\end{array} } \right]$, $\breve{R}_{22,1}=\begin{bmatrix}
-2 &1  &       &0\\
1  &-2  &\ddots \\
           &\ddots    &\ddots  &1\\
 0         &          &1 &-2
\end{bmatrix}$ and $\breve{R}_{22,2}=\begin{bmatrix}
2 &1  &       &0\\
1  &2  &\ddots \\
           &\ddots    &\ddots  &1\\
 0         &          &1 &2 
\end{bmatrix}$. Using Lemma~4 in~\cite{MWPY16:cdc}, we only need to show that $\left(\mathcal{P}_{1} R \mathcal{P}_{1}^{\top}, \mathcal{P}_{1} P\right)$ is controllable. Since $\mathcal{P}_{1} P= \begin{bmatrix}i \\ 0_{2\aleph \times 1} \end{bmatrix}$  and $\mathcal{P}_{1} R \mathcal{P}_{1}^{\top}=
\begin{bmatrix}
0 &\breve{R}_{21}^{\top}  \\
\breve{R}_{21}     &\breve{R}_{22}   
\end{bmatrix}$, according to Lemma~5 in~\cite{MWPY16:cdc}, it suffices to show $(\breve{R}_{22},\; \breve{R}_{21})$ is controllable. According to Lemma~6 in~\cite{MWPY16:cdc}, it suffices to show  that $(\breve{R}_{22,1},\; \begin{bmatrix}
0_{(\aleph-1) \times 1}\\1 
\end{bmatrix}$) and $(\breve{R}_{22,2},\; \begin{bmatrix}
1 \\ 0_{(\aleph-1) \times 1} 
\end{bmatrix}$) are both controllable and that $\breve{R}_{22,1}$ and $\breve{R}_{22,2}$ have no common eigenvalues. Applying Lemma~5 in~\cite{MWPY16:cdc} recursively,  we can easily establish that $(\breve{R}_{22,2},\; \begin{bmatrix}
1 \\ 0_{(\aleph-1) \times 1} 
\end{bmatrix}$) is controllable. Using a similar method, it can be established that $(\breve{R}_{22,1},\; \begin{bmatrix}
0_{(\aleph-1) \times 1}\\1 
\end{bmatrix}$) is controllable. Next we show that $\breve{R}_{22,1}$ and $\breve{R}_{22,2}$ have no common eigenvalues.  It follows from Theorem~2.2 in~\cite{KST99:LAA} that the eigenvalues of $\breve{R}_{22,1}$ are $\lambda_{1}=-2-2\cos(j\pi/(\aleph+1))$, $j=1,2,\cdots,\aleph$, and the eigenvalues of $\breve{R}_{22,2}$ are $\lambda_{2}=2-2\cos(j\pi/(\aleph+1))$, $j=1,2,\cdots,\aleph$. Since $\lambda_{1}<0<\lambda_{2}$, it follows that $\breve{R}_{22,1}$ and $\breve{R}_{22,2}$ have no common eigenvalues. Combining the results above, we conclude that the rank constraint~\eqref{rankconstraint} holds. Therefore, the resulting linear quantum system is strictly stable and generates the $(2\aleph+1)$-mode vacuum state.   Substituting the matrices $P$, $\Gamma$ and $R$ into~\eqref{G} and~\eqref{C}, we obtain the system Hamiltonian $
\hat{H}=\frac{1}{2}\hat{x}^{\top}G\hat{x}=-\sum\limits_{j=1}^{\aleph}\left(\hat{q}_{j}^{2}  + \hat{p}_{j}^{2} \right)+\sum\limits_{j=\aleph+2}^{2\aleph+1}\left(\hat{q}_{j}^{2}  + \hat{p}_{j}^{2} \right)+\sum\limits_{j=1}^{2\aleph}\left(\hat{q}_{j}\hat{q}_{j+1}+\hat{p}_{j}\hat{p}_{j+1} \right)
$, which satisfies the constraint~\ref{constraint1}. 
The coupling vector is given by 
$\hat{L}=C\hat{x}=\hat{q}_{\aleph+1}+i\hat{p}_{\aleph+1}=\sqrt{2}\hat{a}_{\aleph+1}$, 
where $\hat{a}_{\aleph+1}=(\hat{q}_{\aleph+1}+i\hat{p}_{\aleph+1})/\sqrt{2}$ is the annihilation operator for the $(\aleph+1)$th mode. It can be seen that the coupling vector $\hat{L}$ obtained here satisfies the constraint~\ref{constraint2}. Thus, we obtain a desired quantum system that satisfies the constraints~\ref{constraint1} and~\ref{constraint2}, and also generates the $(2\aleph+1)$-mode vacuum state.
\end{proof}

\section{Parametrization} \label{parametrization}
In Proposition~\ref{rmkvacuum}, we have shown that the $(2\aleph+1)$-mode vacuum state can be prepared by a quantum harmonic  oscillator chain subject to the two constraints~\ref{constraint1} and \ref{constraint2}. It is natural to ask if  there exist other pure Gaussian states that can be prepared by quantum harmonic  oscillator chains subject to~\ref{constraint1} and \ref{constraint2}.  The aim of this section is to develop an answer to this question. We will provide a full parametrization of such pure Gaussian states. The main result is given in Theorem~\ref{thm1}. Before providing it, we need several preliminary results. 

\begin{definition}[\cite{HJ12:book}]
A matrix $T$ is called tridiagonal if $T_{jk}=0$ whenever $|j-k|>1$.  
\end{definition}

\begin{definition}[\cite{P98:book}]
A symmetric tridiagonal matrix $T=\begin{bmatrix}
\alpha_{1} &\beta_{1}  &       &0\\
\beta_{1}  &\alpha_{2} &\ddots \\
           &\ddots    &\ddots  &\beta_{n-1}\\
 0         &          &\beta_{n-1} &\alpha_{n} 
\end{bmatrix}$ is said to be unreduced if $\beta_{j}\ne 0$, $j=1,2,\cdots,n-1$.  
\end{definition}

\begin{definition}[\cite{HJ12:book}]
A \emph{Jacobi matrix} is a real symmetric tridiagonal matrix with positive subdiagonal entries. 
\end{definition}

\begin{exam}
$
T_{+}=\begin{bmatrix}
\alpha_{1} &\beta_{1}  &       &0\\
\beta_{1}  &\alpha_{2} &\ddots \\
           &\ddots    &\ddots  &\beta_{n-1}\\
 0         &          &\beta_{n-1} &\alpha_{n} 
\end{bmatrix}$ is a Jacobi matrix if $\alpha_{j}\in\mathbb{R}$, $j=1,2,\cdots,n$ and $\beta_{j}\in\mathbb{R}$, $\beta_{j} >0$, $j=1,2,\cdots,n-1$. 
\end{exam}

\begin{lem}[\cite{P98:book}]\label{jacobiunique} 
Suppose $T_{+}=\mathcal{Q}_{+}^{\top}D\mathcal{Q}_{+}$, where $T_{+}$ is a Jacobi matrix, $\mathcal{Q}_{+}=\begin{bmatrix}
\mathfrak{q}_{1} &\mathfrak{q}_{2} &\cdots &\mathfrak{q}_{n}
\end{bmatrix}$ is a real orthogonal matrix and $D$ is a real diagonal matrix. Then $T_{+}$ and $\mathcal{Q}_{+}$ are  uniquely determined by $D$ and $\mathfrak{q}_{1}$ or by $D$ and $\mathfrak{q}_{n}$. 
\end{lem}

Suppose $T_{+}=\mathcal{Q}_{+}^{\top}D\mathcal{Q}_{+}$. Then given $D$ and $\mathfrak{q}_{1}$, we can use the following iterative algorithm to solve for $T_{+}$ and $\mathcal{Q}_{+}$~\cite[Chapter 7]{P98:book}. 

\RestyleAlgo{boxruled}
\begin{algorithm}[H]$\;$\\
  \caption{Given $D\in \mathbb{R}^{n \times n}$ and $\mathfrak{q}_{1}\in \mathbb{R}^{n \times 1}$, solve for $T_{+}$ and $\mathcal{Q}_{+}$.\label{algo1}}
\textbf{Initialize}: Define $\beta_{0}=0$ and $\mathfrak{q}_{0}=0_{n\times 1}$. Set $j=1$.\\
\textbf{repeat}\\
1. Compute $\alpha_{j}=\mathfrak{q}_{j}^{\top}D\mathfrak{q}_{j}$.  \\
2. Stopping criterion. \textbf{Quit} if $j=n$. \\
3. Compute $\mathfrak{r}_{j}\triangleq D\mathfrak{q}_{j}-\mathfrak{q}_{j}\alpha_{j}-\mathfrak{q}_{j-1}\beta_{j-1}$. \\
4. Compute $\beta_{j}=\norm{\mathfrak{r}_{j}}$.\\
5. Compute $\mathfrak{q}_{j+1}= \mathfrak{r}_{j}/\beta_{j}$.\\
6. Update $j:=j+1$. 
\end{algorithm}

For convenience, we introduce the following notation. 
\begin{itemize}
\item $\textbf{Alg1}_{T_{+}}(D, \mathfrak{q}_{1})$: $\;$ the Jacobi matrix $T_{+}$ obtained from Algorithm~\ref{algo1} for  given  $D$ and $\mathfrak{q}_{1}$.
\item $\textbf{Alg1}_{\mathcal{Q}_{+}}(D, \mathfrak{q}_{1})$: $\;$  the real orthogonal matrix $\mathcal{Q}_{+}$  obtained from Algorithm~\ref{algo1} for  given  $D$ and  $\mathfrak{q}_{1}$.
\end{itemize}

Suppose $T_{+}=\mathcal{Q}_{+}^{\top}D\mathcal{Q}_{+}$. Then  given $D$ and $\mathfrak{q}_{n}$, we can use the following  iterative algorithm to solve for $
T_{+} $ and $\mathcal{Q}_{+} $~\cite[Chapter 7]{P98:book}. 

\RestyleAlgo{boxruled}
\begin{algorithm}[H]$\;$\\
  \caption{ Given $D\in \mathbb{R}^{n \times n}$ and $\mathfrak{q}_{n}\in \mathbb{R}^{n \times 1}$, solve for $T_{+}$ and $\mathcal{Q}_{+}$.\label{algo2}} 
\textbf{Initialize}: Define $\beta_{n}=0$ and $\mathfrak{q}_{n+1}=0_{n\times 1}$. Set $j=n$.\\
\textbf{repeat}\\
1. Compute $\alpha_{j}=\mathfrak{q}_{j}^{\top}D\mathfrak{q}_{j}$. \\
2. Stopping criterion. \textbf{Quit} if $j=1$.\\
3. Compute $\mathfrak{r}_{j-1}\triangleq D\mathfrak{q}_{j}-\mathfrak{q}_{j}\alpha_{j}-\mathfrak{q}_{j+1}\beta_{j}$. \\
4. Compute $\beta_{j-1}=\norm{\mathfrak{r}_{j-1}}$.\\
5. Compute $\mathfrak{q}_{j-1}= \mathfrak{r}_{j-1}/\beta_{j-1}$.\\
6. Update $j:=j-1$. 
\end{algorithm}

For convenience, we introduce the following notation.
\begin{itemize}
\item $\textbf{Alg2}_{T_{+}}(D, \mathfrak{q}_{n})$: $\;$  the Jacobi matrix $T_{+}$ obtained from Algorithm~\ref{algo2} for given $D$ and $\mathfrak{q}_{n}$. 
\item $\textbf{Alg2}_{\mathcal{Q}_{+}}(D, \mathfrak{q}_{n})$: $\;$ the real orthogonal matrix $\mathcal{Q}_{+}$  obtained from Algorithm~\ref{algo2} for given  $D$ and $\mathfrak{q}_{n}$.
\end{itemize}

\begin{rmk}
Algorithm~\ref{algo1} and Algorithm~\ref{algo2} are referred to as \emph{Lanczos algorithms}~\cite{GV96:book}. Note that Algorithm~\ref{algo1} and Algorithm~\ref{algo2} work well under the conditions described in Lemma~\ref{jacobiunique}. However, if we feed an arbitrary real diagonal matrix $D$ and an arbitrary real  unit vector  $\mathfrak{q}_{1}$ into Algorithm~\ref{algo1}, the algorithm may  fail to find a  Jacobi matrix $T_{+}$ and a real orthogonal matrix $\mathcal{Q}_{+}$. 
For example, if $D=I_{n}$ and $\mathfrak{q}_{1}=\begin{bmatrix}1 &0_{1\times(n-1)}\end{bmatrix}^{\top}$, then Algorithm~\ref{algo1} will terminate at the first step since $\beta_{1}=0$. Hence there does not exist a Jacobi matrix $
T_{+} $ and  a real orthogonal matrix $\mathcal{Q}_{+}=\begin{bmatrix}
\mathfrak{q}_{1} &\mathfrak{q}_{2} &\cdots &\mathfrak{q}_{n}
\end{bmatrix}$ such that $T_{+}=\mathcal{Q}_{+}^{\top}D\mathcal{Q}_{+}$ in this case. A similar situation can occur for Algorithm~\ref{algo2}. 
\end{rmk}

To ensure that  Algorithm~\ref{algo1} works, we have the following result. 
\begin{lem}[\cite{GV96:book}]\label{uniqueq1}
Suppose $D\in \mathbb{R}^{n \times n}$ is a real diagonal matrix and $\mathfrak{q}_{1}\in \mathbb{R}^{n \times 1}$ is a real  unit vector. If 
\begin{align*}
\rank\left(\begin{bmatrix}
\mathfrak{q}_{1} &D\mathfrak{q}_{1} &\cdots &D^{n-1}\mathfrak{q}_{1}
\end{bmatrix}\right)=n, 
\end{align*}
then $D$ and $\mathfrak{q}_{1}$  uniquely determine a Jacobi matrix $T_{+}$ and a real orthogonal matrix $\mathcal{Q}_{+}=\begin{bmatrix}
\mathfrak{q}_{1} &\mathfrak{q}_{2} &\cdots &\mathfrak{q}_{n}
\end{bmatrix}$, such that $T_{+}=\mathcal{Q}_{+}^{\top}D\mathcal{Q}_{+}$. In addition,  $T_{+}$ and $\mathcal{Q}_{+}$ can be obtained from Algorithm~\ref{algo1}. 
\end{lem}

To ensure that  Algorithm~\ref{algo2} works, we have a similar result. 
\begin{lem}\label{uniqueqn}
Suppose $D\in \mathbb{R}^{n \times n}$ is a real diagonal matrix and $\mathfrak{q}_{n}\in \mathbb{R}^{n \times 1}$ is a real  unit vector. If 
\begin{align*}
\rank\left(\begin{bmatrix}
\mathfrak{q}_{n} &D\mathfrak{q}_{n} &\cdots &D^{n-1}\mathfrak{q}_{n}
\end{bmatrix}\right)=n, 
\end{align*}
then $D$ and $\mathfrak{q}_{n}$  uniquely determine a Jacobi matrix $T_{+}$ and a real orthogonal matrix $\mathcal{Q}_{+}=\begin{bmatrix}
\mathfrak{q}_{1} &\mathfrak{q}_{2} &\cdots &\mathfrak{q}_{n}
\end{bmatrix}$, such that $T_{+}=\mathcal{Q}_{+}^{\top}D\mathcal{Q}_{+}$.   In addition,  $T_{+}$ and $\mathcal{Q}_{+}$ can be obtained from Algorithm~\ref{algo2}. 
\end{lem}
\begin{proof}
Because $\rank\left(\begin{bmatrix}
\mathfrak{q}_{n} &D\mathfrak{q}_{n} &\cdots &D^{n-1}\mathfrak{q}_{n}
\end{bmatrix}\right)=n$, it follows from Lemma~\ref{uniqueq1} that $D$ and $\mathfrak{q}_{n}$  uniquely determine a Jacobi matrix $\tilde{T}_{+}$ and a real orthogonal matrix $\tilde{\mathcal{Q}}_{+}=\begin{bmatrix}
\mathfrak{q}_{n} & \mathfrak{q} _{n-1} &\cdots & \mathfrak{q} _{1}
\end{bmatrix}$, such that $\tilde{T}_{+}=\tilde{\mathcal{Q}}_{+}^{\top}D\tilde{\mathcal{Q}}_{+}$. Let $\mathcal{P}_{T}=\begin{bmatrix}
0&&&1\\
&&1\\
&\adots\\
1 &&&0
\end{bmatrix}$. Then we have  $\mathcal{P}_{T}\tilde{T}_{+}\mathcal{P}_{T}=\mathcal{P}_{T} \tilde{\mathcal{Q}}_{+}^{\top}D\tilde{\mathcal{Q}}_{+} \mathcal{P}_{T}$. Let  $T_{+}=\mathcal{P}_{T}\tilde{T}_{+}\mathcal{P}_{T}$ and $\mathcal{Q}_{+}=\tilde{\mathcal{Q}}_{+} \mathcal{P}_{T}$. We have $T_{+}=\mathcal{Q}_{+}^{\top}D\mathcal{Q}_{+}$.  It is straightforward to show that $T_{+}$
  is a Jacobi matrix and that $\mathcal{Q}_{+}$ is a real orthogonal matrix with the last column being $\mathfrak{q}_{n}$. The uniqueness of $T_{+}$ and $\mathcal{Q}_{+}$  follows immediately from Lemma~\ref{jacobiunique}. Thus,  $T_{+}$ and $\mathcal{Q}_{+}$ can be obtained from Algorithm~\ref{algo2}.   
\end{proof}

Next we provide our main result which parametrizes the class of pure Gaussian states that can be prepared by quantum harmonic  oscillator chains subject to the constraints~\ref{constraint1} and~\ref{constraint2}.
\begin{thm} \label{thm1}
A $(2\aleph+1)$-mode pure Gaussian state can be prepared by a quantum harmonic  oscillator chain subject to the  constraints~\ref{constraint1} and~\ref{constraint2} if and only if its Gaussian graph matrix can be written as 
\begin{align}
Z=\mathcal{P}_{1}^{\top}\begin{bmatrix}
\bar{z} &0_{1\times 2\aleph}\\
0_{2\aleph\times 1} & \mathcal{Q}^{\top}\mathcal{P}_{2}^{\top} \bar{Z}  \mathcal{P}_{2} \mathcal{Q}
\end{bmatrix}\mathcal{P}_{1},\;\; \bar{Z}=\begin{bmatrix}\tilde{Z}_{1} &&0\\ &\ddots\\0 &&\tilde{Z}_{ \aleph }\end{bmatrix}, \label{thmZ}
\end{align}
where  $ \bar{z}\in \Lambda $, $\tilde{Z}_{j}\in\left\{\begin{bmatrix}\frac{\bar{z}^{2}-1}{2\bar{z}} &\frac{\bar{z}^{2}+1}{2\bar{z}}\\ \frac{\bar{z}^{2}+1}{2\bar{z}} &\frac{\bar{z}^{2}-1}{2\bar{z}} \end{bmatrix},\; \;\begin{bmatrix}\frac{\bar{z}^{2}-1}{2\bar{z}} &-\frac{\bar{z}^{2}+1}{2\bar{z}}\\ -\frac{\bar{z}^{2}+1}{2\bar{z}} &\frac{\bar{z}^{2}-1}{2\bar{z}} \end{bmatrix} \right\}$, $j=1,2,\cdots, \aleph$, $\mathcal{P}_{1}=\begin{bmatrix}
0_{1\times \aleph} &1 &0_{1\times \aleph}\\
I_{\aleph} &0_{\aleph\times 1} &0_{\aleph\times \aleph}\\
0_{\aleph\times \aleph} &0_{\aleph\times 1} &I_{\aleph}
\end{bmatrix}$, $\mathcal{P}_{2}$ is a $2\aleph\times 2\aleph$ permutation matrix, $\mathcal{Q}=\begin{bmatrix}
\mathcal{Q}_{11} &0_{\aleph\times \aleph}\\
0_{\aleph\times \aleph} &\mathcal{Q}_{22} 
\end{bmatrix}$ is a $2\aleph\times 2\aleph$ real orthogonal matrix with 
\begin{align}
\mathcal{Q}_{11}&=\textbf{Alg2}_{\mathcal{Q}_{+}}(\ohill{R}_{11}, \bar{\mathfrak{q}}_{\aleph}\bar{\delta}_{\aleph})\diag\left[\bar{\delta}_{1}, \cdots,\bar{\delta}_{\aleph}\right],\; \bar{\delta}_{j}=\pm 1,  \label{Q11}\\
\mathcal{Q}_{22}&=\textbf{Alg1}_{\mathcal{Q}_{+}}(\ohill{R}_{22}, \tilde{\mathfrak{q}}_{1}\tilde{\delta}_{1})\diag\left[\tilde{\delta}_{1}, \cdots,\tilde{\delta}_{\aleph}\right], \; \tilde{\delta}_{j}=\pm 1, \label{Q22}\\
\ohill{R}_{11}&= \begin{bmatrix}
I_{\aleph} &0_{\aleph\times \aleph}
\end{bmatrix}\mathcal{P}_{2}^{\top}\bar{R}\mathcal{P}_{2}\begin{bmatrix}
I_{\aleph} &0_{\aleph\times \aleph}
\end{bmatrix}^{\top},\label{themR11}\\ 
\ohill{R}_{22}&= \begin{bmatrix}
0_{\aleph\times \aleph} &I_{\aleph}
\end{bmatrix}\mathcal{P}_{2}^{\top}\bar{R}\mathcal{P}_{2}\begin{bmatrix}
0_{\aleph\times \aleph} &I_{\aleph}
\end{bmatrix}^{\top},\label{themR22}\\
\bar{R}&=\diag\left[
r_{1},\;\; 
-r_{1},\;\; 
r_{2},\; \;
-r_{2},\;\; 
\cdots ,\;\; 
r_{\aleph},\; \;
-r_{\aleph}\right], \notag\\
&\text{with}\; r_{j}\in \mathbb{R}, \; \;  r_{j} \ne 0,\;\;  |r_{j}|\ne |r_{k}|\;  \text{ whenever}\;\; j\ne k, \\
\bar{\mathfrak{q}}_{\aleph}&= \pm \frac{ \begin{bmatrix}
I_{\aleph} &0_{\aleph\times \aleph}
\end{bmatrix}\mathcal{P}_{2}^{\top}\wp}{\norm{  \begin{bmatrix}
I_{\aleph} &0_{\aleph\times \aleph}
\end{bmatrix}\mathcal{P}_{2}^{\top}\wp}}, \label{qn}\\
\tilde{\mathfrak{q}}_{1}&=  \pm\frac{ \begin{bmatrix}
0_{\aleph\times \aleph} & I_{\aleph}
\end{bmatrix}\mathcal{P}_{2}^{\top}\wp}{\norm{ \begin{bmatrix}
0_{\aleph\times \aleph} & I_{\aleph}
\end{bmatrix}\mathcal{P}_{2}^{\top}\wp}}, \label{q1}\\
\wp&\in \mathbb{R}^{2 \aleph \times 1} \text{  is a real eigenvector having no zero entries} \notag  \\ 
&\quad \text{ associated with the eigenvalue  $-\frac{1}{\bar{z}}$ of  $\bar{Z}$.} \label{thmwp}
\end{align}
Here $\Lambda \triangleq \Big\{ z \;\;\big|\;\;  z \in \mathbb{C} \;\; \text{and}\;\;  \im(z) >0\Big\}$.
\end{thm} 

\begin{rmk}
If $\bar{z}=i$, we have $\tilde{Z}_{j}=iI_{2}$, $j=1,2,\cdots, \aleph$. It follows that $\bar{Z}=iI_{2\aleph}$. Then we have $Z=iI_{2\aleph+1}$ which corresponds to the Gaussian graph matrix of the $(2\aleph+1)$-mode vacuum state. 
\end{rmk}
\begin{rmk}
If $\bar{z}\ne i$, the vector $\wp$ in~\eqref{thmwp} is of the form $\wp=\begin{bmatrix}\tilde{\wp}_{1}^{\top} &\cdots &
\tilde{\wp}_{\aleph}^{\top} \end{bmatrix}^{\top}$, where
\begin{equation}
\left\{\begin{aligned}
\tilde{\wp}_{j}&=\begin{bmatrix} \tau_{j} \\ -\tau_{j}\end{bmatrix} ,\;\; \tau_{j}\in \mathbb{R},\;\; \tau_{j} \ne 0,\;\;
 \text{if}\;\; \tilde{Z}_{j}=\begin{bmatrix}\frac{\bar{z}^{2}-1}{2\bar{z}} &\frac{\bar{z}^{2}+1}{2\bar{z}}\\ \frac{\bar{z}^{2}+1}{2\bar{z}} &\frac{\bar{z}^{2}-1}{2\bar{z}} \end{bmatrix}, \\
\tilde{\wp}_{j}&=\begin{bmatrix} \tau_{j} \\ \tau_{j}\end{bmatrix},\;\; \tau_{j}\in \mathbb{R},\;\; \tau_{j} \ne 0,\;\; \text{if}\;\; \tilde{Z}_{j}=\begin{bmatrix}\frac{\bar{z}^{2}-1}{2\bar{z}} &-\frac{\bar{z}^{2}+1}{2\bar{z}}\\ -\frac{\bar{z}^{2}+1}{2\bar{z}} &\frac{\bar{z}^{2}-1}{2\bar{z}} \end{bmatrix}.
\end{aligned}\right. \notag
\end{equation}
\end{rmk}

The proof of Theorem~\ref{thm1} is provided in the Appendix. Next we give an example to illustrate Theorem~\ref{thm1}. 
\begin{exam}
Consider a $7$-mode ($\aleph=3$) pure Gaussian state with Gaussian graph matrix given by
\footnotesize
\begin{align}
Z=             i\begin{bmatrix}
\cosh(2\alpha)         & 0       &0               & 0        & 0        & 0 &- \sinh(2\alpha)    \\
    0                   &\cosh(2\alpha) &  0             &0         & 0        &\sinh(2\alpha) &  0   \\
    0                   &0        &\cosh(2\alpha) &0         & -  \sinh(2\alpha)  &  0 &  0    \\
    0                   &0        &0               &\cosh(2\alpha)+\sinh(2\alpha) &0  &  0 &  0   \\
     0                  & 0       &-\sinh(2\alpha)  &0         &\cosh(2\alpha)     &  0      &  0      \\
     0                 &\sinh(2\alpha)   & 0              & 0          &  0        &\cosh(2\alpha) &  0  \\
-   \sinh(2\alpha)       &  0       &    0           &  0        &  0         &0    &\cosh(2\alpha) 
\end{bmatrix}. \label{exampleZ}
\end{align} 
\normalsize
 We already know from~\cite{ZLV15:pra} that this pure Gaussian state can be generated by a quantum harmonic  oscillator chain subject to the two constraints~\ref{constraint1} and~\ref{constraint2}. Next we show that the parametrization given by Theorem~\ref{thm1} successfully includes the Gaussian graph matrix~\eqref{exampleZ} as a special case. In Theorem~\ref{thm1} let us choose 
\begin{align*}
\bar{z}&=i\left(\cosh(2\alpha)+\sinh(2\alpha)\right),\\
\tilde{Z}_{j}&=i\begin{bmatrix}
\cosh(2\alpha) & (-1)^{j} \sinh(2\alpha) \\
(-1)^{j}\sinh(2\alpha)    &\cosh(2\alpha)              
\end{bmatrix}, \;\; j=1,\;2,\;3,\\
\mathcal{P}_{2}&=\begin{bmatrix}
1 &0 &0 &0 &0 &0\\
0 &0 &0 &0 &0 &1\\
0 &1 &0 &0 &0 &0\\
0 &0 &0 &0 &1 &0\\
0 &0 &1  &0 &0 &0\\
0 &0 &0 &1 &0 &0
\end{bmatrix},\\
\bar{R}&=\diag\big[-1-2\sqrt{2}, 1+2\sqrt{2}, -1, 1, -1+2\sqrt{2}, 1-2\sqrt{2}\big],\\
\wp &= \begin{bmatrix}1 &1 & \sqrt{2}  &- \sqrt{2}  & 1  & 1 \end{bmatrix}^{\top},\\
\bar{\mathfrak{q}}_{3}&= \frac{ \begin{bmatrix}
I_{3} &0_{3\times 3}
\end{bmatrix}\mathcal{P}_{2}^{\top}\wp}{\norm{  \begin{bmatrix}
I_{3} &0_{3\times 3}
\end{bmatrix}\mathcal{P}_{2}^{\top}\wp}} = \begin{bmatrix}\frac{1}{2}  &\frac{\sqrt{2}}{2} &\frac{1}{2} \end{bmatrix}^{\top},\\
\tilde{\mathfrak{q}}_{1}&=  \frac{ \begin{bmatrix}
0_{3\times 3} & I_{3}
\end{bmatrix}\mathcal{P}_{2}^{\top}\wp}{\norm{ \begin{bmatrix}
0_{3\times 3} & I_{3}
\end{bmatrix}\mathcal{P}_{2}^{\top}\wp}}= \begin{bmatrix}\frac{1}{2} &-\frac{\sqrt{2}}{2} &\frac{1}{2}  \end{bmatrix}^{\top},\\
\bar{\delta}_{j}&=\tilde{\delta}_{j}= 1, \;\; j=1,\;2,\;3. 
\end{align*}
Then substituting $\bar{R}$ and $\mathcal{P}_{2}$ into~\eqref{themR11} and~\eqref{themR22} yields 
\begin{align*}
\ohill{R}_{11}&= \diag\big[-1-2\sqrt{2},\;   -1, \; -1+2\sqrt{2}\big], \\
\ohill{R}_{22}&= \diag\big[1-2\sqrt{2},\;   1, \; 1+2\sqrt{2}\big]. 
\end{align*}
To solve for $\mathcal{Q}_{11}$ and $\mathcal{Q}_{22}$ in~\eqref{Q11} and~\eqref{Q22}, we need to apply Algorithm~\ref{algo2} and Algorithm~\ref{algo1}, respectively. We find 
$
\mathcal{Q}_{11}=\mathcal{Q}_{22}=\begin{bmatrix}
\frac{1}{2} &-\frac{\sqrt{2}}{2} &\frac{1}{2}\\
-\frac{\sqrt{2}}{2} &0 &\frac{\sqrt{2}}{2}\\
\frac{1}{2}  &\frac{\sqrt{2}}{2} &\frac{1}{2}
\end{bmatrix}$. 
Substituting $\mathcal{P}_{1}$, $\mathcal{Q}=\diag\left[
\mathcal{Q}_{11},\;\mathcal{Q}_{22} 
\right]$, $\mathcal{P}_{2}$, $\bar{z}$ and $\tilde{Z}_{j}$, $j=1,\;2,\;3$, obtained above into~\eqref{thmZ} gives exactly the same $Z$ as~\eqref{exampleZ}. Thus, we conclude that the Gaussian graph matrix~\eqref{exampleZ} is included in the parametrization given by Theorem~\ref{thm1}. 

Using Lemma~\ref{lemnaoki}, we can construct a quantum harmonic  oscillator chain that satisfies the constraints~\ref{constraint1} and~\ref{constraint2} and also generates the pure Gaussian state with the Gaussian graph matrix~\eqref{exampleZ}.  
Let $
R=\mathcal{P}_{1}^{\top}\begin{bmatrix}
0 &\breve{R}_{21}^{\top}\\
\breve{R}_{21} & \mathcal{Q}^{\top}\mathcal{P}_{2}^{\top} \bar{R}  \mathcal{P}_{2} \mathcal{Q}
\end{bmatrix}\mathcal{P}_{1}$,
where $\breve{R}_{21}=\begin{bmatrix}
0_{2 \times 1}\\
\bar{\mathfrak{q}}_{3}^{\top} \begin{bmatrix}I_{3} &0_{3\times 3}
\end{bmatrix}\mathcal{P}_{2}^{\top}\wp \\
\tilde{\mathfrak{q}}_{1}^{\top} \begin{bmatrix}
0_{3\times 3} & I_{3}
\end{bmatrix}\mathcal{P}_{2}^{\top}\wp\\
0_{2\times 1}\end{bmatrix}$, $\Gamma=0$ and $P=i\frac{\cosh(\alpha)-\sinh(\alpha)}{\sqrt{2}}\begin{bmatrix}
0_{3\times 1}\\
1\\
0_{3\times 1}
\end{bmatrix}$ in~\eqref{G} and~\eqref{C}. Then we obtain $R = \begin{bmatrix}

   -1    &2  & 0         &0         &0         &0         &0\\
    2   &-1   & 2        & 0        & 0        & 0         &0\\
    0   & 2   &-1    &2         &0        & 0        & 0\\
         0         &0   & 2        & 0    &2         &0         &0\\
         0        & 0        & 0   & 2   & 1   & 2  & 0\\
         0         &0        & 0        & 0   & 2   & 1   & 2\\
         0         &0        & 0        & 0   &0   & 2  &  1 \end{bmatrix}$. It can be verified that the rank constraint~\eqref{rankconstraint} holds. Therefore, the resulting quantum harmonic  oscillator chain is strictly stable and generates the pure Gaussian state with the Gaussian graph matrix~\eqref{exampleZ}.  The system Hamiltonian is given by 
 $
\hat{H}=\frac{1}{2}\hat{x}^{\top}G\hat{x}=-\frac{1}{2}\sum\limits_{j=1}^{3}\left(\hat{q}_{j}^{2}  + \hat{p}_{j}^{2} \right)+\frac{1}{2}\sum\limits_{j=5}^{7}\left(\hat{q}_{j}^{2}  + \hat{p}_{j}^{2} \right)+2\sum\limits_{j=1}^{6}\left(\hat{q}_{j}\hat{q}_{j+1}+\hat{p}_{j}\hat{p}_{j+1}\right)
$, which satisfies the constraint~\ref{constraint1}. 
The coupling vector is given by 
$\hat{L}=C\hat{x}= 
i\frac{\cosh(\alpha)-\sinh(\alpha)}{\sqrt{2}}\left[-i(\left(\cosh(2\alpha)+\sinh(2\alpha)\right)\hat{q}_{4}+\hat{p}_{4}\right]
= \frac{\cosh(\alpha)+\sinh(\alpha)}{\sqrt{2}} \hat{q}_{4}+i\frac{\cosh(\alpha)-\sinh(\alpha)}{\sqrt{2}} \hat{p}_{4}  
=\cosh(\alpha)\hat{a}_{4}+\sinh(\alpha)\hat{a}_{4}^{\ast}  
 $, which satisfies the constraint~\ref{constraint2}.   
Lastly, we remark that at steady state, the oscillators symmetrically located with respect to the central one  are entangled in pairs. The steady-state entanglement can be measured by the  logarithmic negativity $\mathcal{E}$~\cite{VW02:pra,P05:prl,ASI04:pra}. The pairwise bipartite entanglement values  are given by $\mathcal{E}_{(1,7)}=\mathcal{E}_{(2,6)} =\mathcal{E}_{(3,5)} =2|\alpha|$. For example,  the pairwise bipartite entanglement values for $\alpha=0.5$ are shown in Fig.~\ref{entanglement1}. We also see that the central (fourth) oscillator  is not entangled with the other oscillators. 
  \begin{figure}[htbp]
\begin{center}
\includegraphics[width=7.5cm]{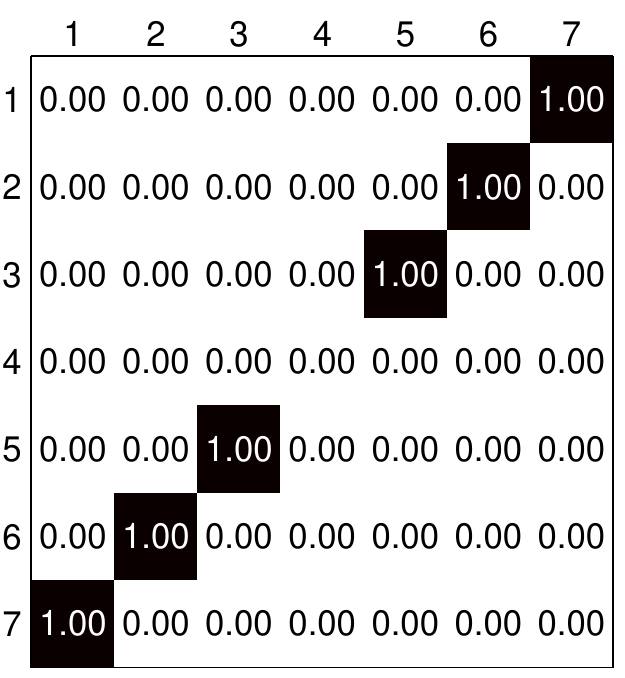}
\caption{
Pairwise bipartite entanglement values in the chain of $7$ quantum harmonic oscillators. Oscillators symmetrically located with respect to the central one  are entangled in pairs. The central (fourth) oscillator is not entangled with the other oscillators.}
\label{entanglement1}
\end{center}
\end{figure}
\end{exam}

\section{Algorithm} \label{Algorithm}
In this section, we will show how to use Theorem~\ref{thm1} to construct useful pure Gaussian states. In particular, according to Theorem~\ref{thm1} and its proof, we outline an algorithm which allows us to find a pure Gaussian state that can be prepared by a quantum harmonic  oscillator chain subject to the two constraints~\ref{constraint1} and \ref{constraint2}. The algorithm consists of six steps. 

\subsection{Algorithm for finding pure Gaussian states}
\noindent
\begin{tabularx}{\textwidth}{|X|}
\hline
\textbf{Step~1.} Choose a complex number $\bar{z}$ from the set $\Lambda$. Choose a  permutation matrix $\mathcal{P}_{2}$. Choose each $\bar{\delta}_{j}$ from the set $\{1,\; -1\}$ for $j=1,2,\cdots, \aleph$. Choose each $\tilde{\delta}_{j}$ from the set  $\{1,\; -1\}$ for $j=1,2,\cdots, \aleph$. Choose $r_{j}$, $j=1,2,\cdots, \aleph$, such that $r_{j}\in \mathbb{R}$, $r_{j} \ne 0$, $|r_{j}|\ne |r_{k}|$  whenever $ j\ne k$.  Let $\bar{R}=\diag\left[
r_{1},\;\; 
-r_{1},\;\; 
r_{2},\; \;
-r_{2},\;\; 
\cdots ,\;\; 
r_{\aleph},\; \;
-r_{\aleph}\right]$.\\ \vspace{0.1cm}

\textbf{Step~2.} Choose each $\tilde{Z}_{j}$ from the set $\left\{ \begin{bmatrix}\frac{\bar{z}^{2}-1}{2\bar{z}} &\frac{\bar{z}^{2}+1}{2\bar{z}}\\ \frac{\bar{z}^{2}+1}{2\bar{z}} &\frac{\bar{z}^{2}-1}{2\bar{z}} \end{bmatrix}, \begin{bmatrix}\frac{\bar{z}^{2}-1}{2\bar{z}} &-\frac{\bar{z}^{2}+1}{2\bar{z}}\\ -\frac{\bar{z}^{2}+1}{2\bar{z}} &\frac{\bar{z}^{2}-1}{2\bar{z}} \end{bmatrix} \right\}$  for $j=1,2,\cdots, \aleph$. Let $\bar{Z}=\diag[\tilde{Z}_{1},\cdots,\tilde{Z}_{ \aleph }]$.\\ 
    \hline 
    \end{tabularx}
    
  \noindent
\begin{tabularx}{\textwidth}{|X|}
\hline  
\textbf{Step~3.} If $\tilde{Z}_{j}=\begin{bmatrix}\frac{\bar{z}^{2}-1}{2\bar{z}} &\frac{\bar{z}^{2}+1}{2\bar{z}}\\ \frac{\bar{z}^{2}+1}{2\bar{z}} &\frac{\bar{z}^{2}-1}{2\bar{z}} \end{bmatrix}$, choose $\tilde{\wp}_{j}=\begin{bmatrix} \tau_{j} \\ -\tau_{j}\end{bmatrix}$, where $\tau_{j}\in \mathbb{R}$ and $\tau_{j} \ne 0$.  Otherwise, choose $\tilde{\wp}_{j}=\begin{bmatrix} \tau_{j} \\ \tau_{j}\end{bmatrix}$, where $\tau_{j}\in \mathbb{R}$ and $\tau_{j} \ne 0$.  Let $\wp=\begin{bmatrix}\tilde{\wp}_{1}^{\top} &\cdots &
\tilde{\wp}_{\aleph}^{\top} \end{bmatrix}^{\top}$.  \\ \vspace{0.1cm}

\textbf{Step~4.} Choose $\bar{\mathfrak{q}}_{\aleph}$ from the set $\left\{\pm \frac{ \begin{bmatrix}
I_{\aleph} &0_{\aleph\times \aleph}
\end{bmatrix}\mathcal{P}_{2}^{\top}\wp}{\norm{  \begin{bmatrix}
I_{\aleph} &0_{\aleph\times \aleph}
\end{bmatrix}\mathcal{P}_{2}^{\top}\wp}}\right\}$. Choose $\tilde{\mathfrak{q}}_{1}$ from the set $\left\{\pm\frac{ \begin{bmatrix}
0_{\aleph\times \aleph} & I_{\aleph}
\end{bmatrix}\mathcal{P}_{2}^{\top}\wp}{\norm{ \begin{bmatrix}
0_{\aleph\times \aleph} & I_{\aleph}
\end{bmatrix}\mathcal{P}_{2}^{\top}\wp}}\right\}$. 
 Calculate $\ohill{R}_{11}= \begin{bmatrix}
I_{\aleph} &0_{\aleph\times \aleph}
\end{bmatrix}\mathcal{P}_{2}^{\top}\bar{R}\mathcal{P}_{2}\begin{bmatrix}
I_{\aleph} &0_{\aleph\times \aleph}
\end{bmatrix}^{\top}$ and $
\ohill{R}_{22}= \begin{bmatrix}
0_{\aleph\times \aleph} &I_{\aleph}
\end{bmatrix}\mathcal{P}_{2}^{\top}\bar{R}\mathcal{P}_{2}\begin{bmatrix}
0_{\aleph\times \aleph} &I_{\aleph}
\end{bmatrix}^{\top}$. \\ \vspace{0.1cm}

\textbf{Step~5.} Feed the real diagonal matrix $\ohill{R}_{11}$ and the real unit vector $\bar{\mathfrak{q}}_{\aleph}\bar{\delta}_{\aleph}$ into Algorithm~\ref{algo2} to obtain the real orthogonal matrix $\textbf{Alg2}_{\mathcal{Q}_{+}}(\ohill{R}_{11}, \bar{\mathfrak{q}}_{\aleph}\bar{\delta}_{\aleph})$. Then calculate $\mathcal{Q}_{11}=\textbf{Alg2}_{\mathcal{Q}_{+}}(\ohill{R}_{11}, \bar{\mathfrak{q}}_{\aleph}\bar{\delta}_{\aleph})\diag\left[\bar{\delta}_{1},\; \cdots,\;\bar{\delta}_{\aleph}\right]$. Feed the real diagonal matrix $\ohill{R}_{22}$ and the real unit vector $\tilde{\mathfrak{q}}_{1}\tilde{\delta}_{1}$ into Algorithm~\ref{algo1} to obtain the real orthogonal matrix $\textbf{Alg1}_{\mathcal{Q}_{+}}(\ohill{R}_{22}, \tilde{\mathfrak{q}}_{1}\tilde{\delta}_{1})$. Then calculate $\mathcal{Q}_{22}=\textbf{Alg1}_{\mathcal{Q}_{+}}(\ohill{R}_{22}, \tilde{\mathfrak{q}}_{1}\tilde{\delta}_{1})\diag\left[\tilde{\delta}_{1},\;\; \cdots,\;\;\tilde{\delta}_{\aleph}\right]$. Let  $\mathcal{Q}=\diag\big[\mathcal{Q}_{11},\;
 \mathcal{Q}_{22} \big]$. \\ \vspace{0.1cm}

\textbf{Step~6.} Calculate the Gaussian graph matrix $Z=\mathcal{P}_{1}^{\top}\begin{bmatrix}
\bar{z} &0_{1\times 2\aleph}\\
0_{2\aleph\times 1} & \mathcal{Q}^{\top}\mathcal{P}_{2}^{\top} \bar{Z}  \mathcal{P}_{2} \mathcal{Q}
\end{bmatrix}\mathcal{P}_{1}$, where $\mathcal{P}_{1}=\begin{bmatrix}
0_{1\times \aleph} &1 &0_{1\times \aleph}\\
I_{\aleph} &0_{\aleph\times 1} &0_{\aleph\times \aleph}\\
0_{\aleph\times \aleph} &0_{\aleph\times 1} &I_{\aleph}
\end{bmatrix}$. After obtaining $Z$, calculate the covariance matrix $V$ of the pure Gaussian state using the formula~\eqref{covariance}. Now we obtain a desired pure Gaussian state with the covariance matrix $V$.\\
\hline
    \end{tabularx}

\begin{rmk}
Once we obtain a pure Gaussian state using the algorithm above, we can immediately find a dissipative quantum harmonic  oscillator chain that generates such a state and also satisfies the constraints~\ref{constraint1} and~\ref{constraint2}. For example, let $
R=\mathcal{P}_{1}^{\top}\begin{bmatrix}
0 &\breve{R}_{21}^{\top}\\
\breve{R}_{21} & \mathcal{Q}^{\top}\mathcal{P}_{2}^{\top} \bar{R}  \mathcal{P}_{2} \mathcal{Q}
\end{bmatrix}\mathcal{P}_{1}$,
where $\breve{R}_{21}=\begin{bmatrix}
0_{(\aleph-1) \times 1}\\
\bar{\mathfrak{q}}_{\aleph}^{\top} \begin{bmatrix}I_{\aleph} &0_{\aleph\times \aleph}
\end{bmatrix}\mathcal{P}_{2}^{\top}\wp \\
\tilde{\mathfrak{q}}_{1}^{\top} \begin{bmatrix}
0_{\aleph\times \aleph} & I_{\aleph}
\end{bmatrix}\mathcal{P}_{2}^{\top}\wp\\
0_{(\aleph-1)\times 1}\end{bmatrix}$, $\Gamma=XRY$ and $P=\begin{bmatrix}
0_{\aleph\times 1}\\
\tau_{p}\\
0_{\aleph\times 1}
\end{bmatrix}$, where 
$\tau_{p}\ne 0$ and $\tau_{p}\in\mathbb{C}$ in~\eqref{G} and~\eqref{C}. Then calculate the matrices $G$ and $C$ using~\eqref{G} and~\eqref{C}, respectively. The resulting linear quantum system with Hamiltonian $\hat{H}=\frac{1}{2}\hat{x}^{\top}G\hat{x}$ and coupling vector $\hat{L}=C\hat{x}$ is strictly stable and   generates the pure Gaussian state. Also, this system is a  quantum  harmonic  oscillator chain that satisfies the two constraints~\ref{constraint1} and~\ref{constraint2}. 
\end{rmk}

\begin{exam}
We use the above algorithm to construct a $7$-mode ($\aleph=3$) pure Gaussian state. We  choose
\begin{align*}
\bar{z}&=0.1+0.45i,\\
\mathcal{P}_{2}&=\begin{bmatrix}
     0     &0     &0     &0     &1     &0\\
     1     &0     &0     &0     &0     &0\\
     0     &0     &1     &0     &0     &0\\
     0     &1     &0     &0     &0     &0\\
     0     &0     &0     &1     &0     &0\\
     0     &0     &0     &0     &0     &1
     \end{bmatrix},\\
\bar{\delta}_{j}&=\tilde{\delta}_{j}=1,\quad j=1,2, 3,\\
\bar{R}&=\diag\left[
-4.2,\;\; 
4.2,\;\; 
-1.5,\; \;
1.5,\;\; 
2,\; \;
-2\right],\\
\tilde{Z}_{1}&=\tilde{Z}_{3}=\begin{bmatrix}
-0.1853 + 1.2838i  &-0.2853 + 0.8338i\\
-0.2853 + 0.8338i  &-0.1853 + 1.2838i
\end{bmatrix},\\
\tilde{Z}_{2}&=\begin{bmatrix}
-0.1853 + 1.2838i  &0.2853 - 0.8338i\\
0.2853 - 0.8338i  &-0.1853 + 1.2838i
\end{bmatrix},\\
\wp&=\begin{bmatrix} 1 &1 &2 &-2 &3 &3\end{bmatrix}^{\top},\\
\bar{\mathfrak{q}}_{3}&=\frac{ \begin{bmatrix}
I_{3} &0_{3\times 3}
\end{bmatrix}\mathcal{P}_{2}^{\top}\wp}{\norm{  \begin{bmatrix}
I_{3} &0_{3\times 3}
\end{bmatrix}\mathcal{P}_{2}^{\top}\wp}}=\frac{1}{3}\begin{bmatrix}
1 &-2 &2
\end{bmatrix}^{\top},\\
\tilde{\mathfrak{q}}_{1}&=\frac{ \begin{bmatrix}
0_{3\times 3} & I_{3}
\end{bmatrix}\mathcal{P}_{2}^{\top}\wp}{\norm{ \begin{bmatrix}
0_{3\times 3} & I_{3}
\end{bmatrix}\mathcal{P}_{2}^{\top}\wp}}=\frac{1}{\sqrt{19}}\begin{bmatrix}
3 &1 &3
\end{bmatrix}^{\top}.
\end{align*}
Then applying Algorithm~\ref{algo2} and Algorithm~\ref{algo1}, respectively, we find  
\begin{align*}
\mathcal{Q}_{11}=   
\begin{bmatrix}
  0.6892    & 0.6433    & 0.3333        \\
    0.6548  &  -0.3561  &  -0.6667         \\
    0.3102  & -0.6778  &  0.6667\\
    \end{bmatrix},\\
    \mathcal{Q}_{22}= \begin{bmatrix}
    0.6882   & 0.7074    &0.1608 \\
         0.2294  & -0.4225   & 0.8769 \\
           0.6882 &  -0.5666  & -0.4530
         \end{bmatrix}. 
\end{align*}
Substituting $\bar{z}$, $\mathcal{P}_{1}$, $\mathcal{P}_{2}$, $\mathcal{Q}=\diag\big[\mathcal{Q}_{11},\;
 \mathcal{Q}_{22} \big]$, and $\bar{Z}=\diag[\tilde{Z}_{1},\tilde{Z}_{2},\tilde{Z}_{ 3}]$ into~\eqref{thmZ}, we obtain the Gaussian graph matrix $Z=X+iY$, where 
 \begin{align*}
 X&=\begin{bmatrix}
   -0.0694  & -0.1581  &    0.0655  &         0     &-0.0451      &0.0831    & -0.1724  \\
   -0.1581  &   -0.0476  &    0.0612  &         0    & -0.0421    &  0.0775   &  -0.1609  \\
    0.0655  &    0.0612  &   -0.4389   &        0    & -0.0218    &  0.0402   &  -0.0834  \\
         0  &         0  &         0   &   0.1000    &       0    &       0    &       0  \\
   -0.0451  &   -0.0421  &   -0.0218   &        0    & -0.4556    & -0.0276    &  0.0574  \\
    0.0831  &    0.0775  &    0.0402   &        0    & -0.0276    &  0.0434    &  0.1174  \\
   -0.1724  &   -0.1609  &   -0.0834   &        0    &  0.0574    &0.1174    & -0.1437  \\
 \end{bmatrix},\\
 Y&=\begin{bmatrix}
    0.9451   &   0.4621   &  -0.1916    &       0     & 0.1318    & -0.2428     & 0.5039  \\
    0.4621   &   0.8813   &  -0.1788    &       0     & 0.1231    & -0.2266    &  0.4703  \\
   -0.1916   &  -0.1788    &  2.0250    &       0     & 0.0638    & -0.1174    &  0.2437  \\
         0   &        0    &       0    &  0.4500     &      0    &       0    &       0  \\
    0.1318    &  0.1231   &   0.0638    &       0     & 2.0738    &  0.0808    & -0.1677  \\
   -0.2428    & -0.2266   &  -0.1174     &      0    &  0.0808    &  0.6153    & -0.3432  \\
    0.5039    &  0.4703   &   0.2437     &      0    & -0.1677    & -0.3432     & 1.1624  \\
 \end{bmatrix}. 
 \end{align*}
 The covariance matrix $V$ can be computed from $Z$ by using the formula~\eqref{covariance}. The pairwise bipartite entanglement between all pairs in the chain can be immediately quantified from its symmetrically ordered covariance matrix $V$ using the logarithmic negativity. The pairwise  bipartite entanglement values are given in Fig.~\ref{entanglement}. We see that every two oscillators (except the central oscillator) are entangled. Hence, this pure Gaussian steady state shows different entanglement properties from that in~\cite{ZLV15:pra}.
  \begin{figure}[htbp]
\begin{center}
\includegraphics[width=8cm]{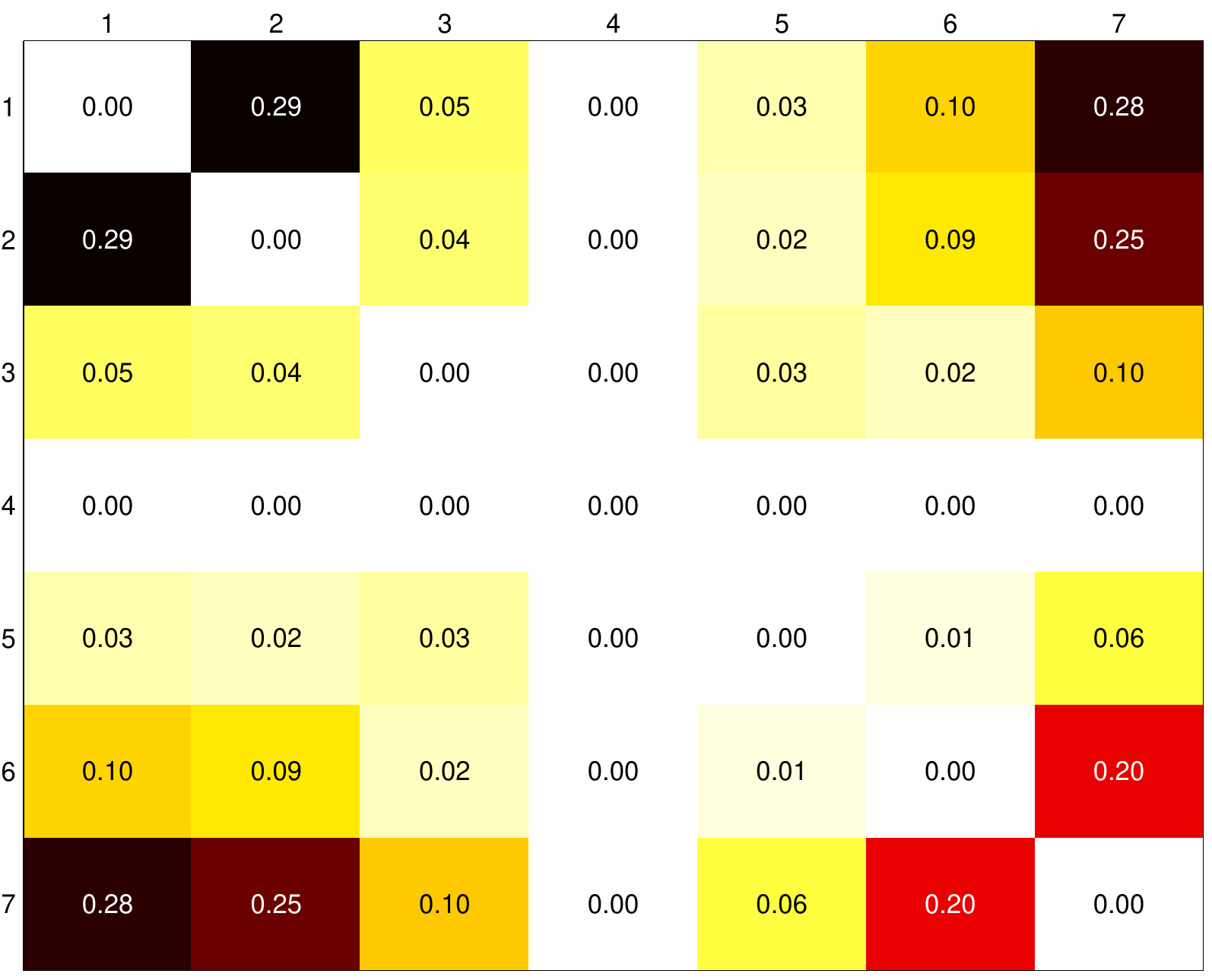}
\caption{
Pairwise bipartite entanglement values in the chain of $7$ quantum harmonic oscillators. Every two oscillators (except the central oscillator) are entangled. The central (fourth) oscillator is not entangled with the other oscillators.}
\label{entanglement}
\end{center}
\end{figure}
 \end{exam}

\begin{exam}
The above algorithm can  be used to find pure Gaussian states with an arbitrary odd mode number. For example, 
Fig.~\ref{entanglement31} shows the pairwise bipartite entanglement values measured by logarithmic negativity of a $31$-mode ($\aleph=15$) pure Gaussian steady state. Due to space limitations, the Gaussian graph matrix and the covariance matrix of this pure Gaussian state are not provided. 
  \begin{figure}[htbp]
\begin{center}
\includegraphics[width=15.5cm]{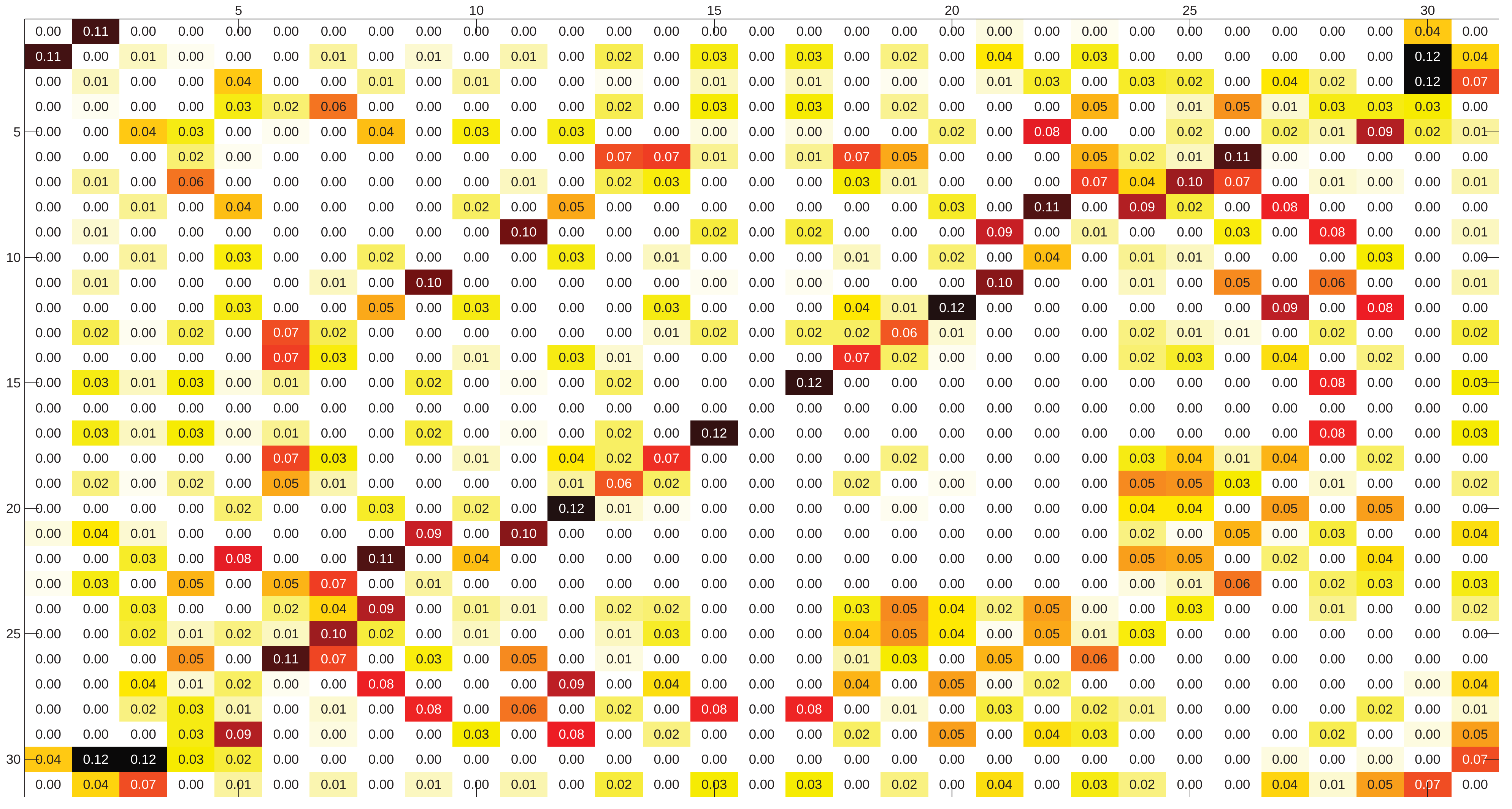}
\caption{
Pairwise bipartite entanglement values in a chain of $31$ quantum harmonic oscillators. The central (sixteenth) oscillator is not entangled with the other oscillators.}
\label{entanglement31}
\end{center}
\end{figure}
\end{exam}

\section{Conclusion} 
In this paper, we consider a chain of $(2\aleph+1)$ quantum harmonic   oscillators subject to constraints. (\textrm{i}) The Hamiltonian $\hat{H}$ is of the form 
$\hat{H}=\sum\limits_{j=1}^{2\aleph+1}\frac{\omega_{j}}{2}\left(\hat{q}_{j}^{2}+\hat{p}_{j}^{2} \right)+\sum\limits_{j=1}^{2\aleph}g_{j}\left(\hat{q}_{j}\hat{q}_{j+1}+\hat{p}_{j}\hat{p}_{j+1} \right)$, where $\omega_{j}\in \mathbb{R}$, $j=1,2,\cdots,2\aleph+1 $, and $g_{j}\in \mathbb{R}$,  $j=1,2,\cdots,2\aleph $.  This type of Hamiltonian describes a set of nearest-neighbour beam-splitter-like interactions. (\textrm{ii})
 Only the central  oscillator of the chain is coupled to the reservoir. That is, the coupling vector $\hat{L}$ is of the form $\hat{L}=c_{1}\hat{q}_{\aleph+1}+c_{2}\hat{p}_{\aleph+1}$,  where $c_{1}\in \mathbb{C}$ and $ c_{2} \in \mathbb{C}$.  Then we derive a sufficient and necessary condition for  a pure Gaussian state   to be prepared in a dissipative quantum harmonic  oscillator chain subject to the above two constraints. These conditions are expressed in terms of a set of constraints on Gaussian graph matrices $Z$. In Section~\ref{Algorithm}, we provide an algorithm for finding those pure Gaussian states by constructing their covariance matrices.  In future work, it would be interesting to investigate the  steady-state entanglement properties in such quantum harmonic  oscillator chains, complementing the work on the  entanglement area law developed in~\cite{AEP02:pra,PEDC05:prl,CEPD06:pra,EP10:arxiv}.  
 
 \ack
 This work was supported by the Australian Research Council (ARC), the Australian Academy of Science, and the Japan Society for the Promotion of Science (JSPS).
 
\section*{Appendix}
In this section, we provide the proof of Theorem~\ref{thm1}. The following preliminary results will be used in the proof. 
\begin{lem} \label{lemunreduced}
Suppose $T=\mathcal{Q}_{{T}}^{\top}D\mathcal{Q}_{T}$, where $T$ is an unreduced   real symmetric tridiagonal matrix, $D$ is a real diagonal matrix and $\mathcal{Q}_{T}=\begin{bmatrix}
\mathfrak{q}_{1} &\mathfrak{q}_{2} &\cdots &\mathfrak{q}_{n}
\end{bmatrix}$ is a  real  orthogonal matrix. Then there exists a diagonal matrix  $J=\diag[\delta_{1},\cdots,\delta_{n}]$, $\delta_{j}=\pm 1$, $j=1,2,\cdots,n$, such that 
\begin{align*}
\mathcal{Q}_{T}=\textbf{Alg1}_{\mathcal{Q}_{+}}(D, \mathfrak{q}_{1}\delta_{1})J=\textbf{Alg2}_{\mathcal{Q}_{+}}(D, \mathfrak{q}_{n}\delta_{n})J.  
\end{align*}
\end{lem}
\begin{proof}
Suppose $
T=\begin{bmatrix}
\alpha_{1} &\beta_{1}  &       &0\\
\beta_{1}  &\alpha_{2} &\ddots \\
           &\ddots    &\ddots  &\beta_{n-1}\\
 0         &          &\beta_{n-1} &\alpha_{n} 
\end{bmatrix}$, where $\alpha_{j}\in\mathbb{R}$, $j=1,2,\cdots,n$ and $\beta_{j} \in\mathbb{R}$, $\beta_{j}\ne 0$, $j=1,2,\cdots,n-1$. Note that $\beta_{j}$ is not necessarily positive.  Let $T_{+}$ be the matrix obtained by replacing each $\beta_{j}$ by $|\beta_{j}|$, $j=1,2,\cdots,n-1$. Then $T_{+}$ is a Jacobi matrix. According to Lemma~7.2.1 in~\cite{P98:book}, there exists a diagonal matrix of the form $J=\diag[\delta_{1},\cdots,\delta_{n}]$, $\delta_{j}=\pm 1$, $j=1,2,\cdots,n$, such that 
\begin{align}
T_{+}=J TJ=J\mathcal{Q}_{T}^{\top}D\mathcal{Q}_{T}J. \label{unreduced11}
\end{align} 
The first column of $\mathcal{Q}_{T}J$ is  $\mathfrak{q}_{1}\delta_{1}$ and the last column of $\mathcal{Q}_{T}J$ is  $\mathfrak{q}_{n}\delta_{n}$. Using Lemma~\ref{jacobiunique},  we   obtain
\begin{numcases}{}
T_{+}=\textbf{Alg1}_{T_{+}}(D, \mathfrak{q}_{1}\delta_{1})=\textbf{Alg2}_{T_{+}}(D, \mathfrak{q}_{n}\delta_{n}), \notag\\
 \mathcal{Q}_{T}J = \textbf{Alg1}_{\mathcal{Q}_{+}}(D, \mathfrak{q}_{1}\delta_{1})= \textbf{Alg2}_{\mathcal{Q}_{+}}(D, \mathfrak{q}_{n}\delta_{n}).  \label{unreduced12}
\end{numcases}  
It follows from~\eqref{unreduced12} that $\mathcal{Q}_{T}=\textbf{Alg1}_{\mathcal{Q}_{+}}(D, \mathfrak{q}_{1}\delta_{1})J=\textbf{Alg2}_{\mathcal{Q}_{+}}(D, \mathfrak{q}_{n}\delta_{n})J$. 
\end{proof}

\begin{lem} \label{lemdiagonal}
Given $\bar{z}\ne i$ and $\bar{z}\in\Lambda$, let $\mathpzc{Z} \in \left\{\begin{bmatrix}\frac{\bar{z}^{2}-1}{2\bar{z}} &\frac{\bar{z}^{2}+1}{2\bar{z}}\\ \frac{\bar{z}^{2}+1}{2\bar{z}} &\frac{\bar{z}^{2}-1}{2\bar{z}} \end{bmatrix},\quad\begin{bmatrix}\frac{\bar{z}^{2}-1}{2\bar{z}} &-\frac{\bar{z}^{2}+1}{2\bar{z}}\\ -\frac{\bar{z}^{2}+1}{2\bar{z}} &\frac{\bar{z}^{2}-1}{2\bar{z}} \end{bmatrix} \right\}$. Suppose $\mathpzc{Z}\mathpzc{R}\mathpzc{Z}=-\mathpzc{R}$, where $\mathpzc{R}$ is a real diagonal matrix. Then $\mathpzc{R}$ is of the form  $\mathpzc{R}=\begin{bmatrix}
\tau  &0\\
0 &-\tau
\end{bmatrix}$, where $\tau\in\mathbb{R}$. 
\end{lem}
\begin{proof}
We write  $\mathpzc{Z}=\begin{bmatrix}
z_{11} &z_{12}\\
z_{12} &z_{11}
\end{bmatrix}$, where $z_{11}=\frac{\bar{z}^{2}-1}{2\bar{z}}$ and $z_{12}=\pm \frac{\bar{z}^{2}+1}{2\bar{z}}$. By assumption,  $\mathpzc{R}$ is a real diagonal matrix, so we write $\mathpzc{R}=\begin{bmatrix}
\tau_{1} &0\\
0 &\tau_{2}
\end{bmatrix}$, where $\tau_{1}\in\mathbb{R}$ and $\tau_{2}\in\mathbb{R}$. Then it follows from $\mathpzc{Z}\mathpzc{R}\mathpzc{Z}=-\mathpzc{R}$ that 
\begin{align*}
\begin{bmatrix}
z_{11} &z_{12}\\
z_{12} &z_{11}
\end{bmatrix}\begin{bmatrix}
\tau_{1} &0\\
0 &\tau_{2}
\end{bmatrix}\begin{bmatrix}
z_{11} &z_{12}\\
z_{12} &z_{11}
\end{bmatrix}&=-\begin{bmatrix}
\tau_{1} &0\\
0 &\tau_{2}
\end{bmatrix},\\
\begin{bmatrix}
z_{11}\tau_{1} &z_{12}\tau_{2}\\
z_{12}\tau_{1} &z_{11}\tau_{2}
\end{bmatrix}\begin{bmatrix}
z_{11} &z_{12}\\
z_{12} &z_{11}
\end{bmatrix}&=-\begin{bmatrix}
\tau_{1} &0\\
0 &\tau_{2}
\end{bmatrix},\\
\begin{bmatrix}
z_{11}^{2}\tau_{1}+z_{12}^{2}\tau_{2} &z_{11}z_{12}\left(\tau_{1}+\tau_{2}\right)\\
z_{11}z_{12}\left(\tau_{1}+\tau_{2}\right) &z_{12}^{2}\tau_{1}+z_{11}^{2}\tau_{2}
\end{bmatrix}&=-\begin{bmatrix}
\tau_{1} &0\\
0 &\tau_{2}
\end{bmatrix}.
\end{align*}
Hence we have $z_{11}z_{12}\left(\tau_{1}+\tau_{2}\right)=0$. Since $\bar{z}\ne i$ and $\bar{z}\in\Lambda$, it is straightforward to show $z_{11}z_{12}\ne 0$. Therefore, we have $\tau_{1}=-\tau_{2}$. This completes the proof. 
 \end{proof}
 
  \begin{lem} \label{lemeigen}
Given $\bar{z}\ne i$ and $\bar{z}\in\Lambda$, let $\mathpzc{Z} \in \left\{\begin{bmatrix}\frac{\bar{z}^{2}-1}{2\bar{z}} &\frac{\bar{z}^{2}+1}{2\bar{z}}\\ \frac{\bar{z}^{2}+1}{2\bar{z}} &\frac{\bar{z}^{2}-1}{2\bar{z}} \end{bmatrix},\quad\begin{bmatrix}\frac{\bar{z}^{2}-1}{2\bar{z}} &-\frac{\bar{z}^{2}+1}{2\bar{z}}\\ -\frac{\bar{z}^{2}+1}{2\bar{z}} &\frac{\bar{z}^{2}-1}{2\bar{z}} \end{bmatrix} \right\}$. Suppose $\xi$ is a real eigenvector of $\mathpzc{Z}$. Then $\xi$ is of the form  $\xi=\begin{bmatrix}
\tau \\
 \pm\tau
\end{bmatrix}$, where $\tau\ne 0$ and $\tau\in \mathbb{R}$. 
\end{lem}
\begin{proof}
We  write  $\mathpzc{Z}=\begin{bmatrix}
z_{11} &z_{12}\\
z_{12} &z_{11}
\end{bmatrix}$, where $z_{11}=\frac{\bar{z}^{2}-1}{2\bar{z}}$ and $z_{12}=\pm \frac{\bar{z}^{2}+1}{2\bar{z}}$.
Suppose $\xi=\begin{bmatrix}
\tau_{1} \\
\tau_{2}
\end{bmatrix}$, where $\tau_{1}\in\mathbb{R}$ and $\tau_{2}\in\mathbb{R}$, is a real eigenvector of $\mathpzc{Z}$, i.e.,  $\begin{bmatrix}
z_{11} &z_{12}\\
z_{12} &z_{11}
\end{bmatrix}\xi=\lambda \xi$. Then we have
\begin{numcases}{}
z_{11}\tau_{1}+z_{12}\tau_{2}=\lambda\tau_{1}, \label{lem11}\\
z_{12}\tau_{1}+z_{11}\tau_{2}=\lambda\tau_{2}. \label{lem12}
\end{numcases}
Adding~\eqref{lem11} and~\eqref{lem12} gives $
\left(z_{11}+z_{12}\right)\left(\tau_{1}+\tau_{2}\right)=\lambda \left(\tau_{1}+\tau_{2}\right)$. 
If $\tau_{1}\ne -\tau_{2}$, then it follows that $\lambda= z_{11}+z_{12}$. 
Substituting this into~\eqref{lem11}, we have $z_{12}\tau_{2}=z_{12}\tau_{1}$. 
Since $\bar{z}\ne i $ and $\bar{z}\in\Lambda$, it is straightforward to show $z_{12}\ne 0$. As a result, we have $\tau_{1}= \tau_{2}$. Therefore, we conclude that  $\tau_{2}=\pm \tau_{1} \ne 0$. This completes the proof. 
\end{proof}

 \noindent\textbf{Proof of Theorem~\ref{thm1}}
 \begin{proof}
\emph{Necessity}. It has been proved in Proposition~\ref{rmkvacuum} that the $(2\aleph+1)$-mode vacuum state can be prepared by a quantum harmonic  oscillator chain subject to the  constraints~\ref{constraint1} and~\ref{constraint2}. So we first show that the corresponding Gaussian graph matrix $Z=iI_{2\aleph+1}$ 
can be written in the form of~\eqref{thmZ}. Let us choose $\bar{z}=i$, $\tilde{Z}_{j}=iI_{2}$, $j=1,2,\cdots, \aleph$, $\mathcal{P}_{2}=I_{2\aleph}$, $\bar{R}=\diag\left[
1,\;\; 
-1,\;\; 
2,\; \;
-2,\;\; 
\cdots ,\;\; 
\aleph,\; \;
-\aleph\right]$, $\wp=\begin{bmatrix}1 &-1 &1 &-1 &\cdots &1 &-1\end{bmatrix}^{\top}$, and $\bar{\delta}_{j}=\tilde{\delta}_{j}=1$, $j=1,2,\cdots, \aleph$. The resulting matrix $Z$ calculated from~\eqref{thmZ} is exactly $Z=iI_{2\aleph+1}$. Therefore, the $(2\aleph+1)$-mode vacuum state is included in the parametrization~\eqref{thmZ} as a special case. 

Next we consider $(2\aleph+1)$-mode non-vacuum pure Gaussian states. Suppose a $(2\aleph+1)$-mode non-vacuum pure Gaussian state is  generated in a $(2\aleph+1)$-mode linear quantum harmonic  oscillator chain subject to the two constraints~\ref{constraint1} and~\ref{constraint2}. We will show that the Gaussian graph matrix $Z$ of this non-vacuum pure Gaussian state can be written in the form of~\eqref{thmZ}.  Since only the $(\aleph+1)$th oscillator of the chain is coupled to the reservoir, it follows from~\eqref{C} that the matrix $P$ is of the form $
P =\begin{bmatrix}
0_{1\times \aleph} &\tau_{p} &0_{1\times \aleph}
\end{bmatrix}^{\top}$, where $\tau_{p}\ne 0$ and $\tau_{p}\in \mathbb{C}$ 
and  the Gaussian graph matrix $Z$ is of the form 
\begin{align}
Z=\begin{bmatrix}
Z_{11} &0_{\aleph\times 1} &Z_{12}\\
0_{1\times \aleph} &\bar{z}  &0_{1\times \aleph}\\
Z_{12}^{\top} &0_{\aleph\times 1} &Z_{22} 
\end{bmatrix}, \label{Zform}
\end{align} 
where $\bar{z}=Z_{\left((\aleph+1),(\aleph+1)\right)}$ is the $\left((\aleph+1),(\aleph+1)\right)$ element of the Gaussian graph matrix $Z$. Since $\im(Z)>0$,  it follows that $\bar{z}\in\Lambda$. 
The constraint~\ref{constraint1} implies that the matrix $G$  in~\eqref{G} satisfies 
\begin{numcases}{}
-XR+\Gamma Y^{-1}=0,\label{GammaXRY}\\
XRX+YRY-\Gamma Y^{-1}X-XY^{-1}\Gamma^{\top}=R,   \label{R}
\end{numcases}
where $R=\begin{bmatrix}
\omega_{1} &g_{1}  &       &0\\
g_{1}  &\omega_{2}  &\ddots \\
           &\ddots    &\ddots  &g_{2\aleph}\\
 0         &          &g_{2\aleph} &\omega_{2\aleph+1} 
\end{bmatrix}$. Since the system  generates the state, we have $g_{j}\ne 0$, $j=1,2,\cdots, 2\aleph$, since otherwise, the system contains an isolated quantum subsystem which is not strictly stable. As a result, $R$ is an \emph{unreduced} real symmetric tridiagonal matrix. 
From~\eqref{GammaXRY}, we have $\Gamma=XRY$. Substituting this into~\eqref{R} yields $YRY-XRX=R$. Combining this with $\Gamma+\Gamma^{\top}=0$ gives $ZRZ=-R$. 
From~\eqref{Zform}, we note that 
\begin{align}
 Z  =\mathcal{P}_{1}^{\top}\begin{bmatrix}
 \bar{z}  &0_{1\times 2\aleph}\\
0_{2\aleph\times 1} &\breve{Z}
\end{bmatrix}\mathcal{P}_{1}, \label{Ztran}
\end{align} where $\mathcal{P}_{1}=\begin{bmatrix}
0_{1\times \aleph} &1 &0_{1\times \aleph}\\
I_{\aleph} &0_{\aleph\times 1} &0_{\aleph\times \aleph}\\
0_{\aleph\times \aleph} &0_{\aleph\times 1} &I_{\aleph}
\end{bmatrix}$ and $\breve{Z}\triangleq\begin{bmatrix}
Z_{11} &Z_{12}\\
Z_{12}^{\top} &Z_{22} 
\end{bmatrix}$. We also have 
\begin{align}
 R= \mathcal{P}_{1}^{\top}\begin{bmatrix}
\omega_{\aleph+1 } &\breve{R}_{21}^{\top}  \\
\breve{R}_{21}     &\breve{R}_{22}   
\end{bmatrix}\mathcal{P}_{1}, \label{Rtran}
\end{align}  where $\breve{R}_{22}=
\left[ {\begin{array}{*{20}c}
\begin{matrix}
\omega_{1} &g_{1}  &       &0\\
g_{1}  &\omega_{2}  &\ddots \\
           &\ddots    &\ddots  &g_{\aleph-1}\\
 0         &          &g_{\aleph-1} &\omega_{\aleph} 
\end{matrix} &\vline &0_{\aleph\times \aleph}   \\
\hline
0_{\aleph\times \aleph} &\vline &\begin{matrix}
\omega_{\aleph+2} &g_{\aleph+2}  &       &0\\
g_{\aleph+2}  &\omega_{\aleph+3}  &\ddots \\
           &\ddots    &\ddots  &g_{2\aleph}\\
 0         &          &g_{2\aleph} &\omega_{2\aleph+1} 
\end{matrix} 
\end{array} } \right]$ and  $\breve{R}_{21}=\begin{bmatrix}
0_{ (\aleph-1)  \times 1} \\ g_{ \aleph } \\ g_{ \aleph+1 } \\ 0_{(\aleph-1)\times 1}
\end{bmatrix}$.  
Recall that $ZRZ=-R$. It follows from~\eqref{Ztran} and~\eqref{Rtran} that
\begin{align*}
\begin{bmatrix}
 \bar{z}  &0_{1\times 2\aleph}  \\
0_{2\aleph\times 1} &\breve{Z} 
\end{bmatrix}
\begin{bmatrix}
\omega_{\aleph+1 } &\breve{R}_{21}^{\top}  \\
\breve{R}_{21}    &\breve{R}_{22}    
\end{bmatrix}\begin{bmatrix}
 \bar{z}  &0_{1\times 2\aleph}  \\
0_{2\aleph\times 1} &\breve{Z} 
\end{bmatrix}=-\begin{bmatrix}
\omega_{\aleph+1 } &\breve{R}_{21}^{\top}  \\
\breve{R}_{21}    &\breve{R}_{22}    
\end{bmatrix}.
\end{align*} 
That is,
\begin{numcases}{}
 \bar{z}  \omega_{ \aleph+1 } \bar{z}  = - \omega_{ \aleph+1 }, \label{pf10}\\
\breve{Z} \breve{R}_{21} \bar{z}  = -\breve{R}_{21}, \label{pf11}\\
\breve{Z} \breve{R}_{22}  \breve{Z}  = -\breve{R}_{22} . \label{pf12}
\end{numcases}

Since $\breve{R}_{22} =\breve{R}_{22}^{\top}$ is a block diagonal matrix, it can be diagonalized as $
\breve{R}_{22} =\begin{bmatrix}
\mathcal{Q}_{11}^{\top}\ohill{R}_{11}\mathcal{Q}_{11}  &0_{\aleph\times \aleph}\\
0_{\aleph\times \aleph} &\mathcal{Q}_{22}^{\top}\ohill{R}_{22}\mathcal{Q}_{22}  
\end{bmatrix}$, where $\mathcal{Q}_{11}$ and $\mathcal{Q}_{22}$ are real orthogonal matrices, and $\ohill{R}_{11}$ and $\ohill{R}_{22}$ are real diagonal matrices. Let $\mathcal{Q}\triangleq \diag\big[
\mathcal{Q}_{11},\;\mathcal{Q}_{22}\big]$. Then we have $
\breve{R}_{22} = \mathcal{Q}^{\top} \begin{bmatrix}
\ohill{R}_{11}  &0_{\aleph\times \aleph}\\
 0_{\aleph\times \aleph} &\ohill{R}_{22}  
\end{bmatrix}\mathcal{Q}$. Let $\ohill{Z}\triangleq \mathcal{Q}\breve{Z} \mathcal{Q}^{\top}$. The equations~\eqref{pf11} and~\eqref{pf12} are transformed into
\begin{numcases}{}
\ohill{Z} \mathcal{Q}  \breve{R}_{21}\bar{z}  = -\mathcal{Q} \breve{R}_{21}, \label{pf21}\\
\ohill{Z} \begin{bmatrix}
\ohill{R}_{11}  &0_{\aleph\times \aleph}\\
 0_{\aleph\times \aleph} &\ohill{R}_{22}  
\end{bmatrix} \ohill{Z}  = - \begin{bmatrix}
\ohill{R}_{11}  &0_{\aleph\times \aleph} \\
 0_{\aleph\times \aleph} &\ohill{R}_{22}  
\end{bmatrix}. \label{pf22}
\end{numcases}

From \eqref{rankconstraint},  we know $(Q,\; P)$ is controllable. Recall that 
$Q =-iRY+Y^{-1}\Gamma=-iRY+Y^{-1}(-YRX)=-RZ 
 =-\mathcal{P}_{1}^{\top} \begin{bmatrix}
\omega_{\aleph+1 } &\breve{R}_{21}^{\top}  \\
\breve{R}_{21}    &\breve{R}_{22}    
\end{bmatrix} \begin{bmatrix}
 \bar{z}  &0_{1\times 2\aleph}\\
0_{2\aleph\times 1} &\breve{Z}
\end{bmatrix}\mathcal{P}_{1}$ and $P  =\big[
0_{1\times \aleph} \;\tau_{p}\; 0_{1\times \aleph}
\big]^{\top}$. 
It follows from Lemma~4 in~\cite{MWPY16:cdc} that $\left(-\begin{bmatrix}
\omega_{\aleph+1 } &\breve{R}_{21}^{\top}  \\
\breve{R}_{21}    &\breve{R}_{22}    
\end{bmatrix} \begin{bmatrix}
 \bar{z}  &0_{1\times 2\aleph}\\
0_{2\aleph\times 1} &\breve{Z}
\end{bmatrix},\;\mathcal{P}_{1}P \right)$ is controllable. Since $\mathcal{P}_{1}P=\begin{bmatrix}
\tau_{p} &0_{1\times 2\aleph}
\end{bmatrix}^{\top}$, it follows from Lemma~5 in~\cite{MWPY16:cdc} that $(-\breve{R}_{22}\breve{Z},\; -\breve{R}_{21}\bar{z})$ is controllable. By Lemma~6 in~\cite{MPW15:arxiv},  $-\breve{R}_{22} \breve{Z}$ is a non-derogatory matrix.   Then following similar arguments as in the proof of Theorem~1 in~\cite{MWPY16:cdc}, we obtain  
\begin{align*}
\ohill{Z}=\mathcal{P}_{2}^{\top}\bar{Z}\mathcal{P}_{2},\quad \bar{Z}=\diag[\tilde{Z}_{1},\cdots,\tilde{Z}_{ \aleph} ],
\end{align*}
where $\mathcal{P}_{2}\in\mathbb{R}^{ 2\aleph \times 2\aleph }$ is a permutation matrix, $\tilde{Z}_{j}\in \left\{\begin{bmatrix}\frac{\bar{z}^{2}-1}{2\bar{z}} &\frac{\bar{z}^{2}+1}{2\bar{z}}\\ \frac{\bar{z}^{2}+1}{2\bar{z}} &\frac{\bar{z}^{2}-1}{2\bar{z}} \end{bmatrix},\quad\begin{bmatrix}\frac{\bar{z}^{2}-1}{2\bar{z}} &-\frac{\bar{z}^{2}+1}{2\bar{z}}\\ -\frac{\bar{z}^{2}+1}{2\bar{z}} &\frac{\bar{z}^{2}-1}{2\bar{z}} \end{bmatrix} \right\}$, $j=1,2,\cdots,   \aleph$.  Then the equations~\eqref{pf21} and~\eqref{pf22} are transformed into  
\begin{numcases}{}
\bar{Z}\mathcal{P}_{2}\mathcal{Q} \breve{R}_{21}\bar{z}  = -\mathcal{P}_{2}\mathcal{Q} \breve{R}_{21},\label{pf31}\\
\bar{Z} \bar{R}\bar{Z}  = - \bar{R},  \label{pf32}
\end{numcases}
where $
\bar{R}\triangleq\mathcal{P}_{2}\begin{bmatrix}
\ohill{R}_{11}  &0_{\aleph\times \aleph}\\
 0_{\aleph\times \aleph} &\ohill{R}_{22}  
\end{bmatrix}\mathcal{P}_{2}^{\top}$ is a real diagonal matrix.  
By assumption, the $(2\aleph+1)$-mode pure Gaussian state generated in the quantum harmonic  oscillator chain is a non-vacuum state. If $\bar{z}= i$, from the analysis above we have $\tilde{Z}_{j}= iI_{2}$, $j=1,2,\cdots,  \aleph$. In this case, it can be further derived that $Z=iI_{2\aleph+1}$ which corresponds to the vacuum state. Hence we have $\bar{z}\ne i$. It follows from~\eqref{pf10} that $\omega_{ \aleph+1 }=0$. According to Lemma~\ref{lemdiagonal}, Eq.~\eqref{pf32} implies that $\bar{R}$ is of the form $
\bar{R}=\diag\left[
r_{1},\;\;
-r_{1},\;\;
r_{2},\;\;
-r_{2},\;\;
\cdots ,\;\;
r_{\aleph},\;\;
-r_{\aleph}\right]$,
where $r_{j}\in \mathbb{R}$, $j=1,2,\cdots, \aleph$. Next we show that $r_{j} \ne 0$, $j=1,2,\cdots, \aleph$, and $|r_{j}|\ne |r_{k}|$ whenever $j\ne k$. Suppose there exists $r_{j} = 0$. Then $-\bar{R}\bar{Z}$ has a  diagonal block $0_{2\times 2}$. In this case, it can be shown that $-\bar{R}\bar{Z}$ is  a derogatory matrix. But we already know that $-\breve{R}_{22} \breve{Z}=-\mathcal{Q}^{\top}\begin{bmatrix}
\ohill{R}_{11}  &0_{\aleph\times \aleph} \\
 0_{\aleph\times \aleph} &\ohill{R}_{22}  
\end{bmatrix}\ohill{Z}  \mathcal{Q}=-\mathcal{Q}^{\top}\mathcal{P}_{2}^{\top}\bar{R}\bar{Z}\mathcal{P}_{2}\mathcal{Q}$ is a non-derogatory matrix.  According to Lemma~5 in~\cite{MPW15:arxiv}, $-\bar{R}\bar{Z}$ must be a non-derogatory matrix. So we reach a contradiction.  Therefore, we have $r_{j} \ne 0$, $j=1,2,\cdots, \aleph$. To show $|r_{j}|\ne |r_{k}|$, for example, we assume 
$|r_{1}|=|r_{2}|\ne 0$. Then it follows from~\eqref{pf32} that
$
\diag\big[
\tilde{Z}_{1},\;\tilde{Z}_{2}\big]\diag[
r_{1}, \;\; -r_{1}, \;\;
r_{2}, \;\;-r_{2}]\diag\big[
\tilde{Z}_{1},\;\tilde{Z}_{2}\big]  = - \diag[
r_{1}, \;\; -r_{1}, \;\;
r_{2}, \;\;-r_{2}]$.
Then we have $\left(\diag[
r_{1}, \;\; -r_{1}, \;\;
r_{2}, \;\;-r_{2}] \diag\big[
\tilde{Z}_{1},\;\tilde{Z}_{2}\big]\right)^{2}=-r_{1}^{2}I_{4}$. Since $r_{1}\ne 0$,  it follows from Lemma~2 in~\cite{MPW15:arxiv}  that $\left(\diag[
r_{1}, \;\; -r_{1}, \;\;
r_{2}, \;\;-r_{2}]\diag\big[
\tilde{Z}_{1},\;\tilde{Z}_{2}\big]\right) $ is diagonalizable and its eigenvalues are either $ir_{1}$ or $-ir_{1}$. In this case,  $\left(\diag[
r_{1}, \;\; -r_{1}, \;\;
r_{2}, \;\;-r_{2}]\diag\big[
\tilde{Z}_{1},\;\tilde{Z}_{2}\big]\right) $  cannot be a non-derogatory matrix. Then it is straightforward to show that the whole matrix $-\bar{R}\bar{Z}$ is not a non-derogatory matrix. Again, we reach a contradiction.  Therefore, we  have  $|r_{j}|\ne |r_{k}|$ whenever $j\ne k$.
Since $\bar{R}=\mathcal{P}_{2}\begin{bmatrix}
\ohill{R}_{11}  &0_{\aleph\times \aleph}\\
 0_{\aleph\times \aleph} &\ohill{R}_{22}  
\end{bmatrix}\mathcal{P}_{2}^{\top}$, it follows that 
\begin{align*}
\ohill{R}_{11}&= \begin{bmatrix}
I_{\aleph} &0_{\aleph\times \aleph}
\end{bmatrix}\mathcal{P}_{2}^{\top}\bar{R}\mathcal{P}_{2}\begin{bmatrix}
I_{\aleph} \\
0_{\aleph\times \aleph}
\end{bmatrix},\\ 
\ohill{R}_{22}&= \begin{bmatrix}
0_{\aleph\times \aleph} &I_{\aleph}
\end{bmatrix}\mathcal{P}_{2}^{\top}\bar{R}\mathcal{P}_{2}\begin{bmatrix}
0_{\aleph\times \aleph}\\
I_{\aleph}
\end{bmatrix}. 
\end{align*}

Let $\wp\triangleq\mathcal{P}_{2}\mathcal{Q}\breve{R}_{21}$. Then it follows from~\eqref{pf31} that $\wp$ is a real eigenvector of $\bar{Z}$ associated with the eigenvalue $-\frac{1}{\bar{z}}$. We next show that $\wp$ has no zero entries. Recall that $\left(-\breve{R}_{22}\breve{Z},-\breve{R}_{21}\bar{z}\right)$ is controllable, i.e., $\left(-\mathcal{Q}^{\top}\mathcal{P}_{2}^{\top}\bar{R}\bar{Z}\mathcal{P}_{2}\mathcal{Q},-\breve{R}_{21}\bar{z}\right)$ is controllable. According to Lemma~4 in~\cite{MWPY16:cdc}, $\left(-\bar{R}\bar{Z},-\mathcal{P}_{2}\mathcal{Q}\breve{R}_{21}\bar{z}\right)$ is controllable. That is, $\left(-\bar{R}\bar{Z},-\wp\bar{z}\right)$ is controllable. Suppose $\wp=\begin{bmatrix}\wp_{1} &\wp_{2} &\cdots &&\wp_{2 \aleph}\end{bmatrix}^{\top}$.  It then  follows from Lemma~6 in~\cite{MWPY16:cdc} that 
$
\left(-\diag\big[
r_{j},\;-r_{j}\big]\tilde{Z}_{j},-\begin{bmatrix}\wp_{2j-1} & \wp_{2j} \end{bmatrix}^{\top}\bar{z}\right)$  is controllable,   $j=1,2,\cdots,\aleph$.   
Hence we have $\begin{bmatrix}\wp_{2j-1} & \wp_{2j} \end{bmatrix}^{\top}\ne 0_{2\times 1}$. 
Since $\wp$ is a real eigenvector of $\bar{Z}$ and $\begin{bmatrix}\wp_{2j-1} & \wp_{2j} \end{bmatrix}^{\top}\ne 0_{2\times 1}$, it follows that $\begin{bmatrix}\wp_{2j-1} & \wp_{2j} \end{bmatrix}^{\top}$ is a real eigenvector of $\tilde{Z}_{j}$. It follows from Lemma~\ref{lemeigen}  that $ \wp_{2j-1} =\pm\wp_{2j}$. Then we have $\wp_{2j-1} \ne 0$ and $\wp_{2j}\ne 0$,   $j=1,2,\cdots,\aleph$. That is, $\wp$ has no zero entries.

Let $\bar{\mathfrak{q}}_{\aleph}$ be the last column of  $\mathcal{Q}_{11}$  and let $\tilde{\mathfrak{q}}_{1}$ be the first column of $\mathcal{Q}_{22}$. Recall that $\breve{R}_{21}=\begin{bmatrix}
0_{(\aleph-1) \times 1}\\ g_{ \aleph } \\   g_{ \aleph+1 }\\ 0_{(\aleph-1)\times 1}
\end{bmatrix}$. So we have $\wp=\mathcal{P}_{2}\mathcal{Q}\breve{R}_{21}= \mathcal{P}_{2}\begin{bmatrix}
\bar{\mathfrak{q}}_{\aleph}g_{\aleph}\\
\tilde{\mathfrak{q}}_{1} g_{ \aleph+1 }
\end{bmatrix}$. Then it follows that
$\bar{\mathfrak{q}}_{\aleph}= \frac{\frac{1}{g_{\aleph}}\begin{bmatrix}
I_{\aleph} &0_{\aleph\times \aleph}
\end{bmatrix}\mathcal{P}_{2}^{\top}\wp}{\norm{\frac{1}{g_{\aleph}} \begin{bmatrix}
I_{\aleph} &0_{\aleph\times \aleph}
\end{bmatrix}\mathcal{P}_{2}^{\top}\wp}}=\pm \frac{ \begin{bmatrix}
I_{\aleph} &0_{\aleph\times \aleph}
\end{bmatrix}\mathcal{P}_{2}^{\top}\wp}{\norm{  \begin{bmatrix}
I_{\aleph} &0_{\aleph\times \aleph}
\end{bmatrix}\mathcal{P}_{2}^{\top}\wp}}$, and  $
\tilde{\mathfrak{q}}_{1}= \frac{\frac{1}{g_{\aleph+1}}\begin{bmatrix}
0_{\aleph\times \aleph} & I_{\aleph}
\end{bmatrix}\mathcal{P}_{2}^{\top}\wp}{\norm{\frac{1}{g_{\aleph+1}}\begin{bmatrix}
0_{\aleph\times \aleph} & I_{\aleph}
\end{bmatrix}\mathcal{P}_{2}^{\top}\wp}}
= \pm\frac{ \begin{bmatrix}
0_{\aleph\times \aleph} & I_{\aleph}
\end{bmatrix}\mathcal{P}_{2}^{\top}\wp}{\norm{ \begin{bmatrix}
0_{\aleph\times \aleph} & I_{\aleph}
\end{bmatrix}\mathcal{P}_{2}^{\top}\wp}}$. 
Recall that $
\breve{R}_{22} =\begin{bmatrix}
\mathcal{Q}_{11}^{\top}\ohill{R}_{11}\mathcal{Q}_{11}  &0_{\aleph\times \aleph}\\
0_{\aleph\times \aleph} &\mathcal{Q}_{22}^{\top}\ohill{R}_{22}\mathcal{Q}_{22}  
\end{bmatrix}$ and both $\mathcal{Q}_{11}^{\top}\ohill{R}_{11}\mathcal{Q}_{11}$ and $\mathcal{Q}_{22}^{\top}\ohill{R}_{22}\mathcal{Q}_{22}$ are \emph{unreduced} real symmetric tridiagonal matrices.  Using Lemma~\ref{lemunreduced}, there exist $\bar{J}=\diag[\bar{\delta}_{1},\cdots,\bar{\delta}_{\aleph}]$, $\bar{\delta}_{j}=\pm 1$ and $\tilde{J}=\diag[\tilde{\delta}_{1},\cdots,\tilde{\delta}_{\aleph}]$, $\tilde{\delta}_{j}=\pm 1$, such that
$\mathcal{Q}_{11}=\textbf{Alg2}_{\mathcal{Q}_{+}}(\ohill{R}_{11}, \bar{\mathfrak{q}}_{\aleph}\bar{\delta}_{\aleph})\bar{J}$, and $\mathcal{Q}_{22} =\textbf{Alg1}_{\mathcal{Q}_{+}}(\ohill{R}_{22}, \tilde{\mathfrak{q}}_{1}\tilde{\delta}_{1})\tilde{J}$. 
Combining all the results above, we conclude that the Gaussian graph matrix $Z$ of the non-vacuum pure Gaussian state  satisfies
\begin{align*}
Z=\mathcal{P}_{1}^{\top}\begin{bmatrix}
\bar{z} &0_{1\times 2\aleph}\\
0_{ 2\aleph \times 1 } & \mathcal{Q}^{\top}\mathcal{P}_{2}^{\top} \bar{Z}  \mathcal{P}_{2} \mathcal{Q}
\end{bmatrix}\mathcal{P}_{1},\; \bar{Z}=\diag[\tilde{Z}_{1},\cdots,\tilde{Z}_{ \aleph }],  
\end{align*}
where $ \bar{z}\in \Lambda $ and $ \bar{z}\ne i $, $\tilde{Z}_{j}\in\left\{\begin{bmatrix}\frac{\bar{z}^{2}-1}{2\bar{z}} &\frac{\bar{z}^{2}+1}{2\bar{z}}\\ \frac{\bar{z}^{2}+1}{2\bar{z}} &\frac{\bar{z}^{2}-1}{2\bar{z}} \end{bmatrix},\; \;\begin{bmatrix}\frac{\bar{z}^{2}-1}{2\bar{z}} &-\frac{\bar{z}^{2}+1}{2\bar{z}}\\ -\frac{\bar{z}^{2}+1}{2\bar{z}} &\frac{\bar{z}^{2}-1}{2\bar{z}} \end{bmatrix} \right\}$, $j=1,2,\cdots, \aleph$, $\mathcal{P}_{1}=\begin{bmatrix}
0_{1\times \aleph} &1 &0_{1\times \aleph}\\
I_{\aleph} &0_{\aleph\times 1} &0_{\aleph\times \aleph}\\
0_{\aleph\times \aleph} &0_{\aleph\times 1} &I_{\aleph}
\end{bmatrix}$, $\mathcal{P}_{2}$ is a $2\aleph\times 2\aleph$ permutation matrix, $\mathcal{Q}=\diag\left[
\mathcal{Q}_{11},\;\mathcal{Q}_{22} 
\right]$ is a real orthogonal matrix with 
\begin{align*}
\mathcal{Q}_{11}&=\textbf{Alg2}_{\mathcal{Q}_{+}}(\ohill{R}_{11}, \bar{\mathfrak{q}}_{\aleph}\bar{\delta}_{\aleph})\diag\left[\bar{\delta}_{1},\; \cdots,\;\bar{\delta}_{\aleph}\right],\; \bar{\delta}_{j}=\pm 1,  \\
\mathcal{Q}_{22}&=\textbf{Alg1}_{\mathcal{Q}_{+}}(\ohill{R}_{22}, \tilde{\mathfrak{q}}_{1}\tilde{\delta}_{1})\diag\left[\tilde{\delta}_{1},\;\; \cdots,\;\;\tilde{\delta}_{\aleph}\right], \; \tilde{\delta}_{j}=\pm 1,  \\
\ohill{R}_{11}&= \begin{bmatrix}
I_{\aleph} &0_{\aleph\times \aleph}
\end{bmatrix}\mathcal{P}_{2}^{\top}\bar{R}\mathcal{P}_{2}\begin{bmatrix}
I_{\aleph} &0_{\aleph\times \aleph}
\end{bmatrix}^{\top}, \\ 
\ohill{R}_{22}&= \begin{bmatrix}
0_{\aleph\times \aleph} &I_{\aleph}
\end{bmatrix}\mathcal{P}_{2}^{\top}\bar{R}\mathcal{P}_{2}\begin{bmatrix}
0_{\aleph\times \aleph} &I_{\aleph}
\end{bmatrix}^{\top}, \\
\bar{R}&=\diag\left[
r_{1},\;\; 
-r_{1},\;\; 
r_{2},\; \;
-r_{2},\;\; 
\cdots ,\;\; 
r_{\aleph},\; \;
-r_{\aleph}\right], \notag\\
&\text{with}\; r_{j}\in \mathbb{R}, \; \;  r_{j} \ne 0,\;\;  |r_{j}|\ne |r_{k}|\;  \text{ whenever}\;\; j\ne k,  \\
\bar{\mathfrak{q}}_{\aleph}&= \pm \frac{ \begin{bmatrix}
I_{\aleph} &0_{\aleph\times \aleph}
\end{bmatrix}\mathcal{P}_{2}^{\top}\wp}{\norm{  \begin{bmatrix}
I_{\aleph} &0_{\aleph\times \aleph}
\end{bmatrix}\mathcal{P}_{2}^{\top}\wp}},  \\
\tilde{\mathfrak{q}}_{1}&=  \pm\frac{ \begin{bmatrix}
0_{\aleph\times \aleph} & I_{\aleph}
\end{bmatrix}\mathcal{P}_{2}^{\top}\wp}{\norm{ \begin{bmatrix}
0_{\aleph\times \aleph} & I_{\aleph}
\end{bmatrix}\mathcal{P}_{2}^{\top}\wp}}, \\
\wp&\in \mathbb{R}^{2 \aleph \times 1} \text{  is a real eigenvector having no zero entries} \notag  \\ 
&\quad \text{ associated with the eigenvalue  $-\frac{1}{\bar{z}}$ of  $\bar{Z}$.}  
\end{align*}
This completes the necessity proof.

\emph{Sufficiency}. We prove the sufficiency by construction. We will construct a quantum harmonic  oscillator chain that satisfies the constraints~\ref{constraint1} and~\ref{constraint2}, and also generates the pure Gaussian state specified by~\eqref{thmZ}.  Since  $\wp$ has no zero entries, it follows from~\eqref{qn} and~\eqref{q1} that $\bar{\mathfrak{q}}_{\aleph}$ and $\tilde{\mathfrak{q}}_{1}$ have no zero entries. Since $r_{j}\in \mathbb{R}$, $r_{j} \ne 0$, and $|r_{j}|\ne |r_{k}|$ whenever $j \ne k$, it follows from~\eqref{themR11} and~\eqref{themR22} that $\ohill{R}_{11}$ and $\ohill{R}_{22}$ are both  real diagonal matrices with distinct nonzero diagonal entries. Using Lemma~6 in~\cite{MWPY16:cdc}, it follows that $(\ohill{R}_{11}, \bar{\mathfrak{q}}_{\aleph}\bar{\delta}_{n})$ and $(\ohill{R}_{22}, \tilde{\mathfrak{q}}_{1}\tilde{\delta}_{1})$ are both controllable. It follows from  Lemma~\ref{uniqueqn} that the matrix $\ohill{R}_{11}$ and the vector $\bar{\mathfrak{q}}_{\aleph}\bar{\delta}_{\aleph}$ uniquely determine a real orthogonal matrix $\textbf{Alg2}_{\mathcal{Q}_{+}}(\ohill{R}_{11}, \bar{\mathfrak{q}}_{\aleph}\bar{\delta}_{\aleph})$. Similarly, it follows from Lemma~\ref{uniqueq1} that the matrix $\ohill{R}_{22}$ and the vector $\tilde{\mathfrak{q}}_{1}\tilde{\delta}_{1}$ uniquely determine  a real orthogonal matrix $\textbf{Alg1}_{\mathcal{Q}_{+}}(\ohill{R}_{22}, \tilde{\mathfrak{q}}_{1}\tilde{\delta}_{1})$. Therefore, the matrices $\mathcal{Q}_{11}$ and $\mathcal{Q}_{22}$ in \eqref{Q11} and~\eqref{Q22} are well defined. 
Let us choose $
R=\mathcal{P}_{1}^{\top}\begin{bmatrix}
0 &\breve{R}_{21}^{\top}\\
\breve{R}_{21} & \mathcal{Q}^{\top}\mathcal{P}_{2}^{\top} \bar{R}  \mathcal{P}_{2} \mathcal{Q}
\end{bmatrix}\mathcal{P}_{1}$,
where  $\breve{R}_{21}=\begin{bmatrix}
0_{(\aleph-1) \times 1}\\
\bar{\mathfrak{q}}_{\aleph}^{\top} \begin{bmatrix}I_{\aleph} &0_{\aleph\times \aleph}
\end{bmatrix}\mathcal{P}_{2}^{\top}\wp \\
\tilde{\mathfrak{q}}_{1}^{\top} \begin{bmatrix}
0_{\aleph\times \aleph} & I_{\aleph}
\end{bmatrix}\mathcal{P}_{2}^{\top}\wp\\
0_{(\aleph-1)\times 1}\end{bmatrix}$, $\Gamma=XRY$ and $P=\begin{bmatrix}
0_{\aleph\times 1}\\
\tau_{p}\\
0_{\aleph\times 1}
\end{bmatrix}$, where $\tau_{p}\ne 0$ and $\tau_{p}\in \mathbb{C}$  in~\eqref{G} and~\eqref{C}. We next show that the resulting linear quantum system with $\hat{H}=\frac{1}{2}\hat{x}^{\top}G\hat{x}$ and $\hat{L}=C\hat{x}$ satisfies the constraints~\ref{constraint1} and~\ref{constraint2}, and also generates the pure Gaussian state with Gaussian graph matrix~\eqref{thmZ}.   Obviously, we have $R=R^{\top}$. Next we show $ZRZ=-R$. We note that 
\begin{align}
ZRZ =&\mathcal{P}_{1}^{\top}\begin{bmatrix}
\bar{z} &0_{1\times 2\aleph}\\
0_{2\aleph\times 1} & \mathcal{Q}^{\top}\mathcal{P}_{2}^{\top} \bar{Z}  \mathcal{P}_{2} \mathcal{Q}
\end{bmatrix}\begin{bmatrix}
0 &\breve{R}_{21}^{\top}\\
\breve{R}_{21} & \mathcal{Q}^{\top}\mathcal{P}_{2}^{\top} \bar{R}  \mathcal{P}_{2} \mathcal{Q}
\end{bmatrix} \notag \\
&\quad \quad \quad \begin{bmatrix}
\bar{z} &0_{1\times 2\aleph}\\
0_{2\aleph\times 1} & \mathcal{Q}^{\top}\mathcal{P}_{2}^{\top} \bar{Z}  \mathcal{P}_{2} \mathcal{Q}
\end{bmatrix}\mathcal{P}_{1}  \notag \\
=&\mathcal{P}_{1}^{\top}\begin{bmatrix}
0 &\bar{z}\breve{R}_{21}^{\top}\\
\mathcal{Q}^{\top}\mathcal{P}_{2}^{\top} \bar{Z}  \mathcal{P}_{2} \mathcal{Q}
\breve{R}_{21} & \mathcal{Q}^{\top}\mathcal{P}_{2}^{\top} \bar{Z} \bar{R} \mathcal{P}_{2} \mathcal{Q}
\end{bmatrix}
\notag \\
&\quad \quad \quad \begin{bmatrix}
\bar{z} & 0_{1\times 2\aleph}\\
0_{2\aleph\times 1} & \mathcal{Q}^{\top}\mathcal{P}_{2}^{\top} \bar{Z}  \mathcal{P}_{2} \mathcal{Q}
\end{bmatrix}\mathcal{P}_{1}  \notag\\
=&\mathcal{P}_{1}^{\top}\begin{bmatrix}
0 &\bar{z}\breve{R}_{21}^{\top}\mathcal{Q}^{\top}\mathcal{P}_{2}^{\top} \bar{Z}  \mathcal{P}_{2} \mathcal{Q}\\
\mathcal{Q}^{\top}\mathcal{P}_{2}^{\top} \bar{Z}  \mathcal{P}_{2} \mathcal{Q}
\breve{R}_{21}\bar{z}  & \mathcal{Q}^{\top}\mathcal{P}_{2}^{\top} \bar{Z}  \bar{R} \bar{Z}  \mathcal{P}_{2} \mathcal{Q}
\end{bmatrix}\mathcal{P}_{1}. \label{suffzrz}
\end{align}

Since $\bar{Z}=\diag[\tilde{Z}_{1},\cdots,\tilde{Z}_{ \aleph }]$, where $\tilde{Z}_{j}\in \left\{\begin{bmatrix}\frac{\bar{z}^{2}-1}{2\bar{z}} &\frac{\bar{z}^{2}+1}{2\bar{z}}\\ \frac{\bar{z}^{2}+1}{2\bar{z}} &\frac{\bar{z}^{2}-1}{2\bar{z}} \end{bmatrix},\quad\begin{bmatrix}\frac{\bar{z}^{2}-1}{2\bar{z}} &-\frac{\bar{z}^{2}+1}{2\bar{z}}\\ -\frac{\bar{z}^{2}+1}{2\bar{z}} &\frac{\bar{z}^{2}-1}{2\bar{z}} \end{bmatrix} \right\}$, $j=1,2,\cdots, \aleph$, and $\bar{R}=\diag\left[
r_{1},\;\; 
-r_{1},\;\; 
r_{2},\; \;
-r_{2},\;\; 
\cdots ,\;\; 
r_{\aleph},\; \;
-r_{\aleph}\right]$, it is straightforward to show that $\bar{Z}  \bar{R} \bar{Z}=- \bar{R} $. From~\eqref{Q11} and~\eqref{Q22}, it can be shown that the last column of $\mathcal{Q}_{11}$ is $\bar{\mathfrak{q}}_{\aleph}$ and the first column of $\mathcal{Q}_{22}$ is $\tilde{\mathfrak{q}}_{1}$. Hence we have
\footnotesize
\begin{align*}
\mathcal{P}_{2} \mathcal{Q}\breve{R}_{21} =\mathcal{P}_{2} \begin{bmatrix}
\mathcal{Q}_{11}  &0_{\aleph\times \aleph}\\
0_{\aleph\times \aleph} &\mathcal{Q}_{22} 
\end{bmatrix}\begin{bmatrix}
0_{(\aleph-1) \times 1}\\
\bar{\mathfrak{q}}_{\aleph}^{\top} \begin{bmatrix}I_{\aleph} &0_{\aleph\times \aleph}
\end{bmatrix}\mathcal{P}_{2}^{\top}\wp \\
\tilde{\mathfrak{q}}_{1}^{\top} \begin{bmatrix}
0_{\aleph\times \aleph} & I_{\aleph}
\end{bmatrix}\mathcal{P}_{2}^{\top}\wp\\
0_{(\aleph-1)\times 1}\end{bmatrix} 
= \mathcal{P}_{2}\begin{bmatrix}
\bar{\mathfrak{q}}_{\aleph}  \left(  \bar{\mathfrak{q}}_{\aleph}^{\top} \begin{bmatrix}I_{\aleph} &0_{\aleph\times \aleph}
\end{bmatrix}\mathcal{P}_{2}^{\top}\wp  \right)
\\
\tilde{\mathfrak{q}}_{1}\left(  \tilde{\mathfrak{q}}_{1}^{\top} \begin{bmatrix}
0_{\aleph\times \aleph} & I_{\aleph}
\end{bmatrix}\mathcal{P}_{2}^{\top}\wp \right)
\end{bmatrix}  \\
=\mathcal{P}_{2}\begin{bmatrix}
\frac{ \begin{bmatrix}
I_{\aleph} &0_{\aleph\times \aleph}
\end{bmatrix}\mathcal{P}_{2}^{\top}\wp \left( \begin{bmatrix}
I_{\aleph} &0_{\aleph\times \aleph}
\end{bmatrix}\mathcal{P}_{2}^{\top}\wp \right)^{\top}}{\left( \norm{ \begin{bmatrix}
I_{\aleph} &0_{\aleph\times \aleph}
\end{bmatrix}\mathcal{P}_{2}^{\top}\wp}\right)^{2}}   
\begin{bmatrix}I_{\aleph} &0_{\aleph\times \aleph}
\end{bmatrix}\mathcal{P}_{2}^{\top}\wp 
\\
\frac{ \begin{bmatrix}
0_{\aleph\times \aleph} & I_{\aleph}
\end{bmatrix}\mathcal{P}_{2}^{\top}\wp\left(  \begin{bmatrix}
0_{\aleph\times \aleph} & I_{\aleph}
\end{bmatrix}\mathcal{P}_{2}^{\top}\wp\right)^{\top}}{\left(\norm{ \begin{bmatrix}
0_{\aleph\times \aleph} & I_{\aleph}
\end{bmatrix}\mathcal{P}_{2}^{\top}\wp}\right)^{2}}
\begin{bmatrix}
0_{\aleph\times \aleph} & I_{\aleph}
\end{bmatrix}\mathcal{P}_{2}^{\top}\wp 
\end{bmatrix}  
= \mathcal{P}_{2} \begin{bmatrix}
   \begin{bmatrix}
I_{\aleph} &0_{\aleph\times \aleph}
\end{bmatrix}\mathcal{P}_{2}^{\top}\wp\\
  \begin{bmatrix}
0_{\aleph\times \aleph} & I_{\aleph}
\end{bmatrix}\mathcal{P}_{2}^{\top}\wp
\end{bmatrix} =\wp. 
\end{align*}
\normalsize 
It then follows that 
\begin{align*}
\mathcal{Q}^{\top}&\mathcal{P}_{2}^{\top} \bar{Z}  \mathcal{P}_{2} \mathcal{Q}
\breve{R}_{21} \bar{z} =\mathcal{Q}^{\top}\mathcal{P}_{2}^{\top} \bar{Z} \wp   \bar{z} 
=-\mathcal{Q}^{\top}\mathcal{P}_{2}^{\top}\wp  \\
=&-\begin{bmatrix}
\mathcal{Q}_{11}^{\top} &0_{\aleph\times \aleph}\\
0_{\aleph\times \aleph} &\mathcal{Q}_{22}^{\top}
\end{bmatrix}\mathcal{P}_{2}^{\top}\wp   
=-\begin{bmatrix}
\mathcal{Q}_{11}^{\top}\begin{bmatrix}I_{\aleph} &0_{\aleph\times \aleph}
\end{bmatrix}\mathcal{P}_{2}^{\top}\wp \\
\mathcal{Q}_{22}^{\top}\begin{bmatrix} 0_{\aleph\times \aleph} &I_{\aleph}
\end{bmatrix}\mathcal{P}_{2}^{\top}\wp
\end{bmatrix}   \\
=&-\begin{bmatrix}
0_{(\aleph-1) \times 1}\\
\bar{\mathfrak{q}}_{\aleph}^{\top} \begin{bmatrix}I_{\aleph} &0_{\aleph\times \aleph}
\end{bmatrix}\mathcal{P}_{2}^{\top}\wp \\
\tilde{\mathfrak{q}}_{1}^{\top} \begin{bmatrix}
0_{\aleph\times \aleph} & I_{\aleph}
\end{bmatrix}\mathcal{P}_{2}^{\top}\wp\\
0_{(\aleph-1)\times 1}\end{bmatrix}  =-\breve{R}_{21},    
\end{align*}
where we use the fact that both $\mathcal{Q}_{11}$ and $\mathcal{Q}_{22}$ are real orthogonal matrices, so their columns are mutually orthogonal.  
Substituting the above equation into~\eqref{suffzrz} and noting that $\bar{Z}  \bar{R} \bar{Z}=- \bar{R} $, we obtain $ZRZ=-R$. Recall that $Z=X+iY$. It then follows that $XRY=-YRX$, i.e., $\Gamma=-\Gamma^{\top}$. Next we show that the rank condition~\eqref{rankconstraint} holds. That is, we need to show $(Q,\;P)$ is controllable. We have 
\small
\begin{align*}
Q&=-iRY+Y^{-1}\Gamma=-RZ\\
&=-\mathcal{P}_{1}^{\top}\begin{bmatrix}
0 &\breve{R}_{21}^{\top}\\
\breve{R}_{21} & \mathcal{Q}^{\top}\mathcal{P}_{2}^{\top} \bar{R}  \mathcal{P}_{2} \mathcal{Q}
\end{bmatrix}\begin{bmatrix}
\bar{z} &0_{1\times 2\aleph}\\
0_{2\aleph\times 1} & \mathcal{Q}^{\top}\mathcal{P}_{2}^{\top} \bar{Z}  \mathcal{P}_{2} \mathcal{Q}
\end{bmatrix}\mathcal{P}_{1}\\
&=-\mathcal{P}_{1}^{\top}\begin{bmatrix}
0 &\breve{R}_{21}^{\top}\mathcal{Q}^{\top}\mathcal{P}_{2}^{\top} \bar{Z}  \mathcal{P}_{2} \mathcal{Q}\\
\breve{R}_{21}  \bar{z} & \mathcal{Q}^{\top}\mathcal{P}_{2}^{\top} \bar{R}    \bar{Z}  \mathcal{P}_{2} \mathcal{Q}  
\end{bmatrix}\mathcal{P}_{1}.
\end{align*}
\normalsize
According to Lemma~4 in~\cite{MWPY16:cdc}, it suffices to  show $\left(\begin{bmatrix}
0 &\breve{R}_{21}^{\top}\mathcal{Q}^{\top}\mathcal{P}_{2}^{\top} \bar{Z}  \mathcal{P}_{2} \mathcal{Q}\\
\breve{R}_{21}  \bar{z} & \mathcal{Q}^{\top}\mathcal{P}_{2}^{\top} \bar{R}    \bar{Z}  \mathcal{P}_{2} \mathcal{Q}  
\end{bmatrix},\;\mathcal{P}_{1}P\right)$ is controllable. Since $\mathcal{P}_{1} P= \begin{bmatrix}\tau_{p} \\ 0_{2\aleph \times 1} \end{bmatrix}$, according to Lemma~5 in~\cite{MWPY16:cdc}, it suffices to  show that $\left(\mathcal{Q}^{\top}\mathcal{P}_{2}^{\top} \bar{R}   \bar{Z}  \mathcal{P}_{2} \mathcal{Q},\;\; \breve{R}_{21}\bar{z} \right)$ is controllable. Again using Lemma~4 in~\cite{MWPY16:cdc}, it suffices to  show  $\left(\bar{R}\bar{Z},\;\; \mathcal{P}_{2} \mathcal{Q} \breve{R}_{21} \bar{z} \right)$ is controllable. That is, we need to  show  $\left(\bar{R}\bar{Z},\;\;\wp \bar{z} \right)$ is controllable. Recall that $\bar{Z}=\diag[\tilde{Z}_{1},\cdots,\tilde{Z}_{ \aleph }]$. Let $\tilde{R}_{j}\triangleq\begin{bmatrix}
r_{j} &0\\
0 &-r_{j}
\end{bmatrix}$, $r_{j}\ne 0$, and $\tilde{\wp}_{j}\triangleq\begin{bmatrix}
\wp_{2j-1}\\ \wp_{2j}
\end{bmatrix}$   where $\wp_{j}$ is the $j$th element of $\wp$. According to Lemma~6 in~\cite{MWPY16:cdc}, it suffices to show $\left(\tilde{R}_{j}\tilde{Z}_{j},\; \tilde{\wp}_{j} \bar{z} \right)$, $j=1,2,\cdots, \aleph$, are all controllable and $\tilde{R}_{j}\tilde{Z}_{j} $, $j=1,2,\cdots, \aleph$, have no common eigenvalues. Since $\wp$  is a real eigenvector having no zero entries associated with the eigenvalue  $-\frac{1}{\bar{z}}$ of  $\bar{Z}$, it follows that $\tilde{\wp}_{j}$ is a real eigenvector having no zero entries associated with the eigenvalue  $-\frac{1}{\bar{z}}$ of  $\tilde{Z}_{j}$. Therefore, we have 
\begin{align*}
&\rank\left(\left[\tilde{\wp}_{j} \bar{z} \;\;\;\;\tilde{R}_{j}\tilde{Z}_{j}\tilde{\wp}_{j}\bar{z} \right]\right)=\rank\left(\left[\tilde{\wp}_{j}\bar{z} \;\;\;\;-\tilde{R}_{j}\tilde{\wp}_{j}\right]\right)\\
&= \rank\left(\begin{bmatrix}
\wp_{2j-1}\bar{z} &-r_{j}\wp_{2j-1}\\
\wp_{2j}\bar{z} &r_{j}\wp_{2j}
\end{bmatrix}\right)=2,\quad j=1,2,\cdots, \aleph. 
\end{align*}
It follows that  $\left(\tilde{R}_{j}\tilde{Z}_{j},\; \tilde{\wp}_{j} \bar{z} \right)$, $j=1,2,\cdots, \aleph$, is controllable. In addition, since $\tilde{Z}_{j}\in\left\{\begin{bmatrix}\frac{\bar{z}^{2}-1}{2\bar{z}} &\frac{\bar{z}^{2}+1}{2\bar{z}}\\ \frac{\bar{z}^{2}+1}{2\bar{z}} &\frac{\bar{z}^{2}-1}{2\bar{z}} \end{bmatrix},\; \;\begin{bmatrix}\frac{\bar{z}^{2}-1}{2\bar{z}} &-\frac{\bar{z}^{2}+1}{2\bar{z}}\\ -\frac{\bar{z}^{2}+1}{2\bar{z}} &\frac{\bar{z}^{2}-1}{2\bar{z}} \end{bmatrix} \right\}$,  it is straightforward to show $(\tilde{R}_{j}\tilde{Z}_{j})^{2}=-\tilde{R}_{j}^{2}=-r_{j}^{2}I_{2}$. Since $r_{j}\ne 0$, using Lemma~2 in~\cite{MPW15:arxiv}, it follows that the matrix $ \tilde{R}_{j}\tilde{Z}_{j} $ is diagonalizable and its eigenvalues are either $ r_{j} i$ or $-r_{j} i$. Since $|r_{j}|\ne |r_{k}|$ whenever $ j\ne k$, hence  $\tilde{R}_{j}\tilde{Z}_{j} $, $j=1,2,\cdots, \aleph$, have no common eigenvalues. By Lemma~6 in~\cite{MWPY16:cdc}, we have established that $\left(\bar{R} \bar{Z}, \;\;\wp \bar{z}  \right)$ is controllable. Then we conclude that $(Q,P)$ is controllable. Hence the resulting linear quantum system is strictly stable and  generates the pure Gaussian state with  Gaussian graph matrix~\eqref{thmZ}. 
Finally, for the matrix $R$,  using~\eqref{themR11} and~\eqref{themR22}, we have 
\begin{align*}
&\mathcal{Q}^{\top}\mathcal{P}_{2}^{\top} \bar{R}  \mathcal{P}_{2} \mathcal{Q}\\
=&   \begin{bmatrix}
\mathcal{Q}_{11}^{\top} &0_{\aleph\times \aleph} \\
0_{\aleph\times \aleph} &\mathcal{Q}_{22}^{\top}
\end{bmatrix} \begin{bmatrix}
\ohill{R}_{11} &0_{\aleph\times \aleph} \\
0_{\aleph\times \aleph} &\ohill{R}_{22}
\end{bmatrix} \begin{bmatrix}
\mathcal{Q}_{11} &0_{\aleph\times \aleph} \\
0_{\aleph\times \aleph} &\mathcal{Q}_{22}
\end{bmatrix}\\
= & \begin{bmatrix}
\mathcal{Q}_{11}^{\top} \ohill{R}_{11} \mathcal{Q}_{11}  &0_{\aleph\times \aleph}  \\
0_{\aleph\times \aleph} &\mathcal{Q}_{22}^{\top} \ohill{R}_{22} \mathcal{Q}_{22} 
\end{bmatrix}. 
\end{align*}
It follows from~\eqref{Q11} and~\eqref{Q22} that $\mathcal{Q}_{11}^{\top} \ohill{R}_{11} \mathcal{Q}_{11}$ and $\mathcal{Q}_{22}^{\top} \ohill{R}_{22} \mathcal{Q}_{22}$ are unreduced real symmetric tridiagonal matrices. It is straightforward to show that the chosen $R$ is an unreduced real symmetric tridiagonal matrix. Substituting $R$, $\Gamma$ and $P$ into~\eqref{G} and~\eqref{C}, we obtain the system Hamiltonian $\hat{H}=\frac{1}{2}\hat{x}^{\top}G\hat{x}=\frac{1}{2}\hat{x}^{\top}\begin{bmatrix}R &0_{(2\aleph+1)\times (2\aleph+1)} \\0_{(2\aleph+1)\times (2\aleph+1)}  &R \end{bmatrix}\hat{x}$, which satisfies the first constraint~\ref{constraint1}. In addition, the system–-reservoir coupling vector $\hat{L}$    is given by 
$\hat{L}= C\hat{x} =P^{\top}\left[-Z\;\; I_{2\aleph+1}\right]\hat{x} = -\tau_{p}\bar{z}\hat{q}_{\aleph+1}+\tau_{p}\hat{p}_{\aleph+1}$, which satisfies the second constraint~\ref{constraint2}. Thus the resulting system satisfies the constraints~\ref{constraint1} and~\ref{constraint2}, and also generates the pure Gaussian state specified by~\eqref{thmZ}. This completes  the sufficiency proof. 
\end{proof}

\section*{References}

\end{document}